\documentclass{paper}
\usepackage[T1]{fontenc}
\usepackage[utf8]{inputenc}

\usepackage{amsmath,amsfonts,amssymb,amsthm,bbold,bbm}
\usepackage{pgf}
\usepackage{comment}
\usepackage{tikz-cd}
\usetikzlibrary{hobby}

\usetikzlibrary{arrows}
\usetikzlibrary{cd}

\newtheorem{theorem}{Theorem}[section]
\newtheorem{proposition}{Proposition}[section]
\newtheorem{corollary}{Corollary}[section]
\newtheorem{remark}{Remark}[section]
\newtheorem{lemma}{Lemma}[section]
\newtheorem{definition}{Definition}[section]

\usepackage[a4paper,margin=0mm,left=30mm,right=30mm,top=25mm,bottom=25mm, headsep=10mm,footskip=15mm,headheight=30mm]{geometry}
\usepackage{hyperref}

\newcommand{\RR}{\mathbb R}

\renewcommand{\b}{\boldsymbol} \newcommand{\ten}{\boldsymbol}
\renewcommand{\tilde}{\widetilde} 

\newcommand{\map}{\mathcal M}
\newcommand{\bone}{{\b 1}} 
\newcommand{\dvec}{{\b d}}
\newcommand{\hvec}{{\b h}}
\newcommand{\xvec}{{\b x}}
\newcommand{\cvec}{{\b c}}
\newcommand{\yvec}{{\b y}}
\newcommand{\vvec}{{\b v}}
\newcommand{\wvec}{{\b w}}
\newcommand{\zvec}{{\b z}}
\newcommand{\R}{{\mathbb{R}}}
\newcommand{\Rn}{\R^n}
\newcommand{\Rnn}{\R^{n\times n}}
\newcommand{\Rnnn}{\R^{n\times n\times n}}

\newcommand{\norm}[1]{\|#1\|}

\begin{document}
\title{A framework for second order eigenvector centralities and clustering coefficients}

\author{Francesca Arrigo\thanks{Department of Mathematics and Statistics, University of Strathclyde, GX11H Glasgow (UK)}
\and Desmond J. Higham\thanks{School of Mathematics, University of Edinburgh, EH93FD Edinburgh (UK)}
\and Francesco Tudisco\thanks{School of Mathematics, GSSI Gran Sasso Science Institute, 67100 L'Aquila (Italy)}
}

\maketitle

\begin{abstract}
We propose  and analyse 
a general tensor-based framework
for incorporating 
second order
features into network measures.
This approach allows us to combine 
traditional 
pairwise links with
information
that 
records whether triples of nodes are involved in 
wedges or triangles.
Our treatment covers classical spectral methods and recently proposed cases from the literature,  but we  
also identify many interesting extensions. 
In particular, we define a mutually-reinforcing (spectral)
version of the 
 classical clustering coefficient.
 The underlying  
object of study is a 
constrained nonlinear 
eigenvalue problem 
associated with a cubic tensor.
Using recent results from 
nonlinear Perron--Frobenius theory, we 
establish existence and uniqueness under appropriate conditions,
and show that the new spectral measures can be computed  
efficiently with a nonlinear power method.
To illustrate the added value 
of the new formulation, we analyse the measures on a class of synthetic networks.
We also give computational results
on centrality and link prediction 
for real-world 
networks.   
\\[.5em]
\textbf{\sffamily AMS subject classifications.}
15A69, 
05C65, 
91D30, 
05C81, 
05C82, 
37C25 
\\[.5em]
\textbf{\sffamily Keywords.} Clustering coefficient, Eigenvector centrality, Higher order network analysis, Tensor, Hypergraph, Perron-Frobenius theory, Link prediction.
\end{abstract}

\maketitle

\section{Introduction and motivation}
The classical paradigm in 
network science is to analyse a complex system by focusing on  
pairwise interactions; that is, by studying lists of nodes and edges.
However, it is now apparent that many important features arise through larger groups of nodes
acting together \cite{benson2016higher}. For example, 
  the \emph{triadic closure} principle 
  from the social sciences 
   suggests that connected node triples, or triangles,
    are important building blocks \cite{bianconi2014triadic,EG19,estrada2015predicting,GHP12}.
     Of course, there is a sense in which 
     many algorithms
     in network science \emph{indirectly} go beyond pairwise interactions by considering traversals around the network.         However,  recent work 
           \cite{B19,BASJK18,EG19, IPBL19} has shown that there is 
            benefit in \emph{directly} taking account of higher-order neighbourhoods when 
            designing algorithms and models. 
             
             From the point of view of algebraic topology, higher-order relations coincide with different homology classes and the idea of exploring connections of higher-order in networks is analogous to the idea of forming a filtered cell complex in topological data analysis \cite{edelsbrunner2010computational}.
In a similar manner to point clouds, complex networks modeling various type of interactions (such as social, biological, communication or food networks) have an intrinsic higher-order organization \cite{benson2016higher}.
So, efficiently accounting for higher-order topology can allow more robust and effective quantification of nodal importance in various senses \cite{C08,OPTH17}.

 Our aim here is to develop and analyse 
a general framework
for incorporating 
second order
features; see Definition~\ref{def:map}. 
This takes the form of a 
constrained nonlinear 
eigenvalue problem 
associated with a  nonlinear mapping defined in terms of a square matrix and a cubic tensor.
For specific parameter choices we recover  
both standard and recently proposed network measures as special cases.
We also 
construct 
many interesting new alternatives.
In this eigenproblem-based setting, the network 
measures naturally incorporate \emph{mutual reinforcement}; 
important objects are those that interact with many other important objects.  
The classic PageRank algorithm
\cite{Gleich15} is perhaps the best known example of such a measure.
Within this setting, in Definition~\ref{def:spec_clus}
 we define for the first time 
a mutually reinforcing version of the 
classical 
\emph{Watts-Strogatz Clustering Coefficient}
\cite{WS98}; here we give extra weight to nodes that form 
triangles with 
nodes that are themselves involved in important triangles.  
We show that our general 
framework can be studied 
using recently developed tools from 
nonlinear Perron--Frobenius theory.
As well as deriving  
existence and uniqueness results
we show that
these measures  
are computable 
via a 
nonlinear extension of the power method; see Theorem~\ref{thm:theory}.
 
The manuscript is organized as follows.
 In section~\ref{sec:bgr} we 
 summarize relevant existing work on spectral measures in network science. 
 Section~\ref{sec:model} sets out a general framework for combining first and second order  
information through a tensor-based nonlinear eigenvalue problem.
 We also give several specific examples in order to show how standard 
 measures can be generalized by including second order terms.
  In section~\ref{sec:theory} we study theoretical and practical issues.  
 Section~\ref{sec:as} illustrates the effect of using second order 
  information through a theoretical analysis on a specific class of networks.
  In section~\ref{sec:numerical} we test the new framework on real large scale networks in the context of centrality assignment and link prediction.
  Conclusions are provided in section~\ref{sec:conc}.

\section{Background and related work}
\label{sec:bgr}

\subsection{Notation}

A {\it network} or {\it graph} $G=(V,E)$ is defined as a pair of sets: nodes $V = \{1,2,\ldots, n\}$ and edges $ E \subseteq V\times V$ among them. 
We assume the graph to be undirected, so that for all $(i,j)\in E$ it also holds that $(j,i)\in E$, unweighted, so that all connections in the network have the same ``strength", and connected, so that it is possible to reach any node in the graph from any other node by following edges. 
We further assume for simplicity that the graph does not contain self-loops, i.e., edges that point from a node to itself.

A graph may  be represented via its 
\emph{adjacency matrix}, $A=(A_{ij})\in\mathbb{R}^{n\times n}$, 
where 
$A_{ij} = 1$ if $(i,j)\in E$ and 
$A_{ij} = 0$ otherwise.
Under our assumptions, this matrix will be symmetric, binary and irreducible. 
We write $G_A$ to denote the graph associated with the adjacency matrix $A$.

We let $\bone \in\mathbb{R}^{n}$ denote the vector with all
components equal to $1$ and $\bone_i\in\mathbb{R}^n$ denote the $i$th vector of the standard basis of $\mathbb{R}^n$.

\subsection{Spectral centrality measures}
A centrality measure 
quantifies the importance of each node by assigning to it a nonnegative value.
This assignment must be invariant under graph isomorphism, meaning that relabeling the nodes does not affect the values they are assigned. 
We focus here 
on {\it degree centrality} and a family of centrality measures that can be described via an eigenproblem involving the adjacency matrix. This latter family includes as special cases \textit{eigenvector centrality} and \textit{PageRank}. 

The {\it degree centrality} of a node is found by simply counting the number of neighbours 
that it possesses; so 
node $i$ is assigned the value 
$d_i$, where $\dvec = A\bone$. 
Degree centrality treats all connections equally; it does not 
take account of the importance of those neighbours.
By contrast 
\emph{eigenvector centrality} 
is based on a recursive relationship where node $i$ is assigned a value
$x_i \ge 0$ such that $\xvec$ is proportional to $A \xvec$.
We will describe this type of measure as 
\emph{mutually reinforcing}, because it gives extra credit to nodes that 
 have more important neighbours.
Under our assumption that $A$ is irreducible, 
the eigenvector centrality
measure $\xvec$ corresponds to the Perron--Frobenius eigenvector of $A$.
We note that this measure was popularized in the 1970s 
by researchers in the social sciences, \cite{Newmanbook}, but can be traced back to
algorithms used 
in the 19th century
for ranking 
chess players, \cite{chess19}. 
For our purposes, it is useful to consider a general class of  
eigenvector based measures of the form 
\begin{equation}\label{eq:eig_linear}
\xvec \geq 0 \quad \text{such that} \quad  M\xvec = \lambda\, \xvec,
\end{equation}
where $M \in\mathbb{R}^{n\times n}$
is defined in terms of the adjacency matrix $A$.
For example, we may use  the adjacency matrix itself, $M=A$, or the \emph{PageRank matrix}  
\begin{equation}
M = c AD^{-1} + (1-c)\vvec\bone^T,
\label{eq:pgmat}
\end{equation}
with $c\in (0,1)$, $\vvec\geq 0$ such that $\norm{\vvec}_1 = 1$ and $D$ the diagonal matrix such that $D_{ii} = d_i$.  
With this  second choice, the eigenvector solution of \eqref{eq:eig_linear}  is the {\it PageRank} vector \cite{Gleich15}. 

\subsection{Watts-Strogatz clustering coefficient} 
The Watts-Strogatz clustering coefficient was used in~\cite{WS98} 
to quantify an aspect of  transitivity for each node. 
To define this coefficient, we use 
$
\triangle(i) = (A^3)_{ii}/2
$
to denote 
the number of {\it unoriented} triangles involving node $i$. 
Note that node $i$ is involved in exactly $d_i(d_i-1)/2$ \emph{wedges} centred at $i$, 
that is, paths of the form $hij$ where $h,i,j$ are distinct.
Hence node $i$ can be involved in at most  
$d_i(d_i-1)/2$ triangles. 
The {\it local Watts--Strogatz clustering coefficient} of node $i$ is defined as the fraction of wedges that are closed into triangles:
\begin{equation}\label{eq:WSCC}
    c_i = 
    \begin{cases}
    \frac{2\triangle(i)}{d_i(d_i-1)} & \text{if } d_i\geq 2 \\
    0 & \text{otherwise.}
    \end{cases}
\end{equation}
It is easy to see that $c_i\in[0,1]$ with $c_i=0$ if node $i$ does not participate in any triangle and $c_i=1$ if node $i$ has not left any wedges unclosed.

Related to this measure of transitivity for nodes there are two 
network-wide versions; the {\it average clustering coefficient} 
\begin{equation*}\label{eq:aWSCC}
\overline{C} 
= \frac{1}{n}\sum_{i=1}^nc_i = \frac{2}{n}\sum_{i: d_i\geq 2} \frac{\triangle(i)}{d_i(d_i-1)}
\end{equation*}
and the \textit{global clustering coefficient} or {\it graph transitivity}~\cite{luce1949method} 
\begin{equation*}\label{eq:gWSCC}
    \widehat{C} = \frac{6|K_3|}{\sum_i d_i(d_i-1)},
\end{equation*}
where $|K_3|$ is the number of unoriented triangles in the network and the multiplicative factor of 6 comes from the fact that each triangle closes six wedges, i.e., the six ordered pairs of edges in the triangle.
This latter measure has been observed to typically take values between $0.1$ and $0.5$ for real world networks;  see~\cite{girvan2002}.  
The global and average clustering coefficients have been 
found to capture meaningful features and have found several applications \cite{mcgraw2005clustering,newman2001structure}; 
however, they may behave rather differently for certain classes of networks \cite{estrada2016local}. 
In this work we focus on the local measure defined in  \eqref{eq:WSCC}.
Beyond social network analysis, this index has  
found application, for example,  in machine learning pipelines, where nodes features are employed to detect outliers~\cite{lafond2014anomaly} or to inform role discovery~\cite{ahmed2018learning,henderson2012rolx},
in epidemiology, where efficient 
vaccination strategies are needed \cite{EnKa18},  
and in psychology~\cite{bearman2004suicide}, where it is desirable to identify  at-risk 
individuals.

We see from (\ref{eq:WSCC}) that the 
Watts--Strogatz clustering coefficient 
may be viewed as a second order equivalent of degree centrality 
in the sense that it is not mutually reinforcing---a node is not given any extra credit for 
forming triangles with well-clustered nodes.
In Definition~\ref{def:spec_clus} below we show how  a 
mutually reinforcing clustering coefficient 
can be defined.

\section{General eigenvector model}\label{sec:model}

To incorporate second order information, 
given a tensor $\ten T\in\Rnnn$ and a parameter $p\in\R$ we define the operator 
$\ten T_p:\Rn\to\Rn$ that maps the vector $\xvec\in\Rn$ to the vector entrywise defined as 
\begin{equation}
\ten T_p(\xvec)_i = \sum_{j,k=1}^n\ten T_{ijk}\ \mu_p(x_j,x_k),
\label{eq:Tp}
\end{equation}
where $\mu_p(a,b)$ is the {\it power} (or {\it binomial}) {\it mean} 
\[
\mu_p(a,b) = \left(\frac{|a|^p+|b|^p}{2}\right)^{1/p}.
\]
Recall that the following well known properties hold for $\mu$: 
i) $\lim_{p\to 0}\mu_p(a,b)  = \sqrt{|ab|}$ is the geometric mean; 
ii) $\mu_{-1}(a,b) = 2(|x|^{-1}+|y|^{-1})^{-1}$ is the harmonic mean;   iii) $\lim_{p\to+\infty}\mu_p = \max\{|a|,|b|\}$ is the maximum function; whereas $\lim_{p\to-\infty}\mu_p =  \min\{|a|,|b|\}$ is the minimum.

We may then define the following nonlinear network operator, and associated 
spectral centrality measure, 
which combines first and second order interactions.

\begin{definition}\label{def:map}
Let $\alpha\in\R$ be such that $0\leq \alpha \leq 1$, let $p\in \RR$ and let $M\in\Rnn$ and $\ten T\in\Rnnn$ be an entrywise nonnegative square matrix and an entrywise nonnegative cubic tensor associated with the network, respectively. Define $\map:\RR^n\to\RR^n$ as
\begin{equation}\label{eq:map}
\map(\xvec) = \alpha M\xvec +(1-\alpha) \ten T_p(\xvec).
\end{equation}
Then the
corresponding
\emph{first and second order eigenvector centrality} of node $i$ is given by $x_i \ge 0$, 
where   $\xvec$ solves the constrained nonlinear eigenvalue problem  
\begin{equation}\label{eq:eig_gen}
\xvec \geq 0 \quad \text{such that} \quad  
 \map(\xvec) = \lambda\,\xvec.
 \end{equation}
\end{definition}

If we set $\alpha = 1$
in
(\ref{eq:eig_gen})
then only 
first order interactions are considered, and 
we return to the 
classical eigenvector centrality measures discussed in section~\ref{sec:bgr}.
Similarly, with $\alpha = 0$ only second order interactions are relevant.  
 
 In the next subsection we discuss specific choices for 
 $M$ and $\ten T$.
 
We also note that in order for the measure in Definition~\ref{def:map}
to be well defined, there must exist a unique solution to the problem 
(\ref{eq:eig_gen}). We consider this issue in section~\ref{sec:theory}.

\subsection{Specifying  $M$ and $\ten T$}\label{ssec:choice_MT}

In Definition~\ref{def:map}, the matrix $M$ 
should encode information about the first order (edge)
interactions, with 
the tensor $\ten T$ representing the triadic relationships among node triples, that is, second order interactions. 

Useful choices of $M$ are therefore the adjacency matrix or 
the PageRank matrix 
(\ref{eq:pgmat}).
Another viable choice, which we will use in some of the numerical experiments, 
is a rescaled version of the adjacency matrix $M = AD^{-1}$, which we will refer to as
the  \textit{random walk matrix}.

We now consider some choices for the tensor $\ten T$ to represent second order interactions. 

\vspace{.5em}

\textbf{Binary triangle tensor.} 
Perhaps the simplest choice of second order tensor is
\begin{equation}\label{eq:TB}
(\ten T_B)_{ijk}  = \begin{cases} 1 & \text{if }i,j,k \text{ form a triangle}\\
0 & \text{otherwise.}
\end{cases}
\end{equation}
As discussed, for example, in \cite{schank2005finding},  we can build ${\ten T_B}$ with worst case computational complexity of $O(n^3)$ or $O(m^{3/2})$, where $n$ is the number of nodes in the network and $m$ is the number of edges. Moreover, in \cite{benson2015tensor} the authors construct the triangles tensor of four large real-world networks {\sc (Email EUAll, soc Epinions1, wiki Talk, twitter combined)} and observe that the number of non-zero entries in ${\ten T_B}$ is $O(6m)$. Note also that this tensor is closely related to the matrix $A\circ A^2$, 
where $\circ$ denotes the componentwise product
(also called the Hadamard or Schur product),
as shown in \eqref{eq:TA}.

It can be easily verified that, regardless of the choice of $p$, $\big((\ten T_B)_p(\bone)\big)_i = (A^3)_{ii} = 2\triangle(i)$ for all $i\in V$.  

\vspace{.5em}

\textbf{Random walk triangle tensor.} A ``random walk" normalization of the tensor $\ten T_B$ in \eqref{eq:TB}, which will be denoted by $\ten T_{W}\in\Rnnn$,  is entrywise defined as
\begin{equation}\label{eq:Tw}
(\ten T_W)_{ijk}  = \begin{cases} \frac{1}{\triangle(j,k)} & \text{if }i,j,k \text{ form a triangle}\\
0 & \text{otherwise,}
\end{cases}
\end{equation}
where $\triangle(j,k) = (A\circ A^2)_{jk}$ is the number of triangles
involving the edge $(j,k)$.  
This is reminiscent of the random walk matrix $M_{ij} = (AD^{-1})_{ij}=\delta_{ij\in E}/d_j$ (here $\delta$ denotes the Kronecker delta) and this is the reason behind the choice of the name.

\vspace{.5em}

\textbf{Clustering coefficient triangle tensor.} An alternative normalization in \eqref{eq:TB} 
gives 
\begin{equation}\label{eq:Tc}
(\ten T_C)_{ijk} = \begin{cases} \frac{1}{d_i(d_i-1)} & \text{if }i,j,k \text{ form a triangle} \\ 
0 & \text{otherwise.}
\end{cases}
\end{equation}
Note that, if $i,j,k$ form a triangle, then $d_i\geq 2$ and 
hence $(\ten T_C)_{ijk}$ is well defined.
This tensor incorporates 
information 
that is not used 
in (\ref{eq:TB}) and (\ref{eq:Tw})---the number of transitive relationships 
that each node could be potentially involved in---while also 
accounting for the second order structure actually present. 
We refer to (\ref{eq:Tl}) as the 
clustering coefficient triangle tensor because for any $p$ 
we have $({\ten T_C})_p(\bone) = \cvec$, the Watts--Strogatz clustering coefficient vector. We will
return to this property in subsection \ref{ssec:spectral_CC}.

\vspace{.5em}

\textbf{Local closure triangle tensor.} 
The \emph{local closure coefficient} 
\cite{yin2019local}
of node $i$ is defined as
\begin{equation}\label{eq:local_closure}
h_i = \frac{2\triangle(i)}{w(i)},
\end{equation}
where 
\begin{equation}\label{eq:local_closure}
w(i) = \sum_{j\in N(i)}d_j - d_i = \sum_{j\in N(i)}(d_j-1)
\end{equation}
is the number of paths of length two originating from node $i$, and $N(i)$ is the set of neighbours of node $i$.
We may also write $\wvec = A\dvec - \dvec = A^2\bone - A\bone$.
The following result, which is an  
immediate consequence of the definition of $w(i)$, shows that we may 
assume 
$w(i)\neq 0$ when dealing with real-world networks.

\begin{proposition}\label{prop:zero_w}
Let $G = (V,E)$ be an unweighted, undirected and connected graph. 
Then $w(i) = 0$ if and only if all neighbours of node $i$ have degree
equal to one.
Further, if $w(i) = 0$ for some $i$ then 
$G$ is either a path graph with two nodes or 
    a star graph with $n\geq 3$ nodes having $i$ as its centre. 
\end{proposition}

We then define the 
local closure triangle tensor as
\begin{equation}\label{eq:Tl}
(\ten T_L)_{ijk} = \begin{cases} \frac{1}{w(i)} & \text{if }i,j,k \text{ form a triangle}\\
0 & \text{otherwise.}
\end{cases}
\end{equation}
It is easily checked that  ${(\ten T_L)}_p(\bone) = \hvec $ for all $p$. 

\vspace{.5em}

Next, we briefly discuss the main differences, for the purposes of this work, among these four tensorial network representations.

The binary triangle tensor (\ref{eq:TB}) and 
random walk triangle tensor (\ref{eq:Tw})
provide no information concerning the wedges involving each node, and hence 
the consequent  potential for triadic closure. Indeed,  networks that have very different structures from the viewpoint of potential and actual transitive relationships are treated alike. For example, consider the two networks in row $(a)$ of Figure~\ref{fig:multi_toy_networks}, where solid lines are used to represent the actual edges in the network.
The two networks are represented by the same 
tensors in the case of  
(\ref{eq:TB}) 
and 
(\ref{eq:Tw}), but are not equivalent from the viewpoint of transitive relationships. 
Indeed, by closing wedges following the principle underlying the Watts-Strogatz clustering coefficient, in the network on the left 
node $1$ could participate in five more triangles, whilst in the graph on the right it could participate in only two more. 
These are highlighted in Figure \ref{fig:multi_toy_networks}, row $(a)$, using dashed lines. 
On the other hand, the clustering coefficient triangle tensor defined in \eqref{eq:Tc} encodes in its entries the ``potential" for triadic closure of node 1; indeed, for the network on the left it holds that $(\ten T_C)_{123} = (\ten T_C)_{132} = 1/12$, while these entries are  $(\ten T_C)_{123} = (\ten T_C)_{132} =1/6$ for the network on the right. 
These values show that there is a potential for node 1 to be involved in respectively $12$ and $6$ directed triangles.
\begin{figure}
\begin{center} 
\begin{tikzpicture}[scale=.9]
\draw (2.5*360/5: 3.5cm) node (a) {\Large $(a)$};
    \draw (360/5: 1.5cm) node[scale=.7,circle,draw=black,ultra thick,fill=blue!30](1){1};
    \draw (2*360/5: 1.5cm) node[scale=.7,circle,draw=black,ultra thick](2){2};
    \draw (3*360/5: 1.5cm) node[scale=.7,circle,draw=black,ultra thick](3){3};
    \draw (4*360/5: 1.5cm) node[scale=.7,circle,draw=black,ultra thick](4){4};
    \draw (5*360/5: 1.5cm) node[scale=.7,circle,draw=black,ultra thick](5){5};
\path[-, ultra thick] (1)edge[] node[]{} (2) (2)edge[] node[]{} (3) (3)edge[] node[]{}(1)
    (1)edge[] node[]{} (4) (1)edge[] node[]{} (5);
    \path[dashed,thick] (2)edge[] node[]{} (4) (2)edge[] node[]{} (5) (3)edge[] node[]{} (4) (3)edge[] node[]{}(5)
    (4)edge[] node[]{} (5);
    \end{tikzpicture}
    \hspace{2cm}
    \begin{tikzpicture}[scale=.9]
    \draw (360/5: 1.5cm) node[scale=.7,circle,draw=black,ultra thick,fill=blue!30](1){1};
    \draw (2*360/5: 1.5cm) node[scale=.7,circle,draw=black,ultra thick](2){2};
    \draw (3*360/5: 1.5cm) node[scale=.7,circle,draw=black,ultra thick](3){3};
    \draw (4*360/5: 1.5cm) node[scale=.7,circle,draw=black,ultra thick](4){4};
    \draw (5*360/5: 1.5cm) node[scale=.7,circle,draw=black,ultra thick](5){5};
    \path[-, ultra thick] (1)edge[] node[]{} (2) (2)edge[] node[]{} (3) (3)edge[] node[]{}(1)
    (1)edge[] node[]{} (4) (4)edge[] node[]{} (5); 
    \path[dashed,thick] (2)edge[] node[]{} (4) (3)edge[] node[]{} (4) (4)edge[] node[]{} (5);
    \end{tikzpicture}\\[2em]
\begin{tikzpicture}[scale=.9]
    \draw (360/5: 1.5cm) node[scale=.7,circle,draw=black,ultra thick,fill=blue!30](1){1};
    \draw (2*360/5: 1.5cm) node[scale=.7,circle,draw=black,ultra thick](2){2};
    \draw (3*360/5: 1.5cm) node[scale=.7,circle,draw=black,ultra thick](3){3};
    \draw (4*360/5: 1.5cm) node[scale=.7,circle,draw=black,ultra thick](4){4};
    \draw (5*360/5: 1.5cm) node[scale=.7,circle,draw=black,ultra thick](5){5};
\path[-, ultra thick] (1)edge[] node[]{} (2) (2)edge[] node[]{} (3) (3)edge[] node[]{}(1)
    (1)edge[] node[]{} (4) (1)edge[] node[]{} (5);
    \draw (2.5*360/5: 3.5cm) node (b) {\Large $(b)$};
    \end{tikzpicture}
    \hspace{2cm}
    \begin{tikzpicture}[scale=.9]
    \draw (360/5: 1.5cm) node[scale=.7,circle,draw=black,ultra thick,fill=blue!30](1){1};
    \draw (2*360/5: 1.5cm) node[scale=.7,circle,draw=black,ultra thick](2){2};
    \draw (3*360/5: 1.5cm) node[scale=.7,circle,draw=black,ultra thick](3){3};
    \draw (4*360/5: 1.5cm) node[scale=.7,circle,draw=black,ultra thick](4){4};
    \draw (5*360/5: 1.5cm) node[scale=.7,circle,draw=black,ultra thick](5){5};
    \path[-, ultra thick] (1)edge[] node[]{} (2) (2)edge[] node[]{} (3) (3)edge[] node[]{}(1)
    (1)edge[] node[]{} (4) (4)edge[] node[]{} (5); 
    \path[dashed,thick] (1)edge[] node[]{} (5);
    \end{tikzpicture}
\end{center} 
\caption{Example networks with the same number of edges (solid) and triangles. Row $(a)$, left: node 1 can be involved in five more undirected triangles according to the principle underlying the Watts--Strogatz clustering coefficient. These are formed using the dashed edges. Row $(a)$, right: node 1 can only be involved in two more, formed using the dashed edges.  Row $(b)$, left: node 1 cannot be involved in any more triangles, according to the principle underlying the local closure coefficient. Row $(b)$, right: node 1 can only be involved in one more, formed using the dashed edge.}\label{fig:multi_toy_networks}
\end{figure}
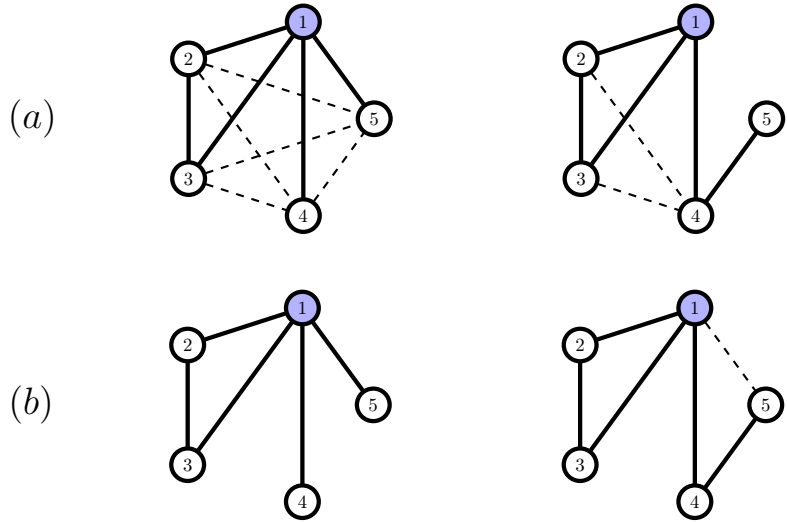

The local closure triangle tensor defined in \eqref{eq:Tl} encodes another type of triadic closure property---the potential of a node to become involved in triangles by connecting to nodes that are at distance two from it. 
In the networks depicted in Figure~\ref{fig:multi_toy_networks}, row $(b)$,  it is clear that no
such triangles can be formed in the network on the left, while there is one that could be formed in the graph on the right (dashed edge). 
For the entries of the associated tensor $\ten T_L$,  the left network in  row $(b)$ of
Figure~\ref{fig:multi_toy_networks}  has $(\ten T_L)_{123} = (\ten T_L)_{132} = 1/2$, and indeed node 1 is participating in both possible directed triangles 
that can be formed according to the principals of local closure. The network on the right  has
$(\ten T_L)_{123} = (\ten T_L)_{132} = 1/3$.

\vspace{.5em}

\subsection{The linear cases: $\alpha = 1$ or $p=1$}
The map $\map$ defined in \eqref{eq:map} becomes linear for particular choices of $p$ and $\alpha$. 
One case arises when $\alpha =1$, whence it reduces to a standard matrix-vector product,
$
\map(\xvec) = M\xvec
$,
and \eqref{def:map} boils down to a linear eigenvector problem \eqref{eq:eig_linear}. 
Using the particular choices of $M$ described in the previous subsection, it then follows that our model includes as a special case standard eigenvector centrality and PageRank centrality. 

Now let $\alpha\in[0,1)$ and $p=1$. 
Then the mapping $\ten T_p:\RR^n \to \RR^n$ also becomes linear; indeed, entrywise it
becomes 
\[
\ten T_1(\xvec)_i = \frac 12 \sum_{j,k=1}^n\ten T_{ijk}x_k+\ten T_{ijk}x_j =\frac 12 \Big\{ \sum_{j=1}^n (\sum_{k=1}^n \ten T_{ikj})x_j +  \sum_{j=1}^n (\sum_{k=1}^n \ten T_{ijk}) x_j \Big\} 
\]
and $\ten T_1(\b x)$ reduces to the product between the vector $\b x$  and the matrix with entries $\frac 12 (\sum_k \ten T_{ijk}+\ten T_{ikj})$. 
In particular, if the tensor $\ten T$ is symmetric with respect to the second and third modes, i.e.\ $\ten T_{ijk} = \ten T_{ikj}$ for all $j,k$, it follows that 
\[
\ten T_1(\xvec)_i = \sum_{j=1}^n (\sum_{k=1}^n \ten  T_{ijk}) x_j\, .
\]
Note that this is the case for all the tensors defined in subsection~\ref{ssec:choice_MT}. 

We now explicitly compute $(\sum_k \ten T_{ijk})$ for some of the tensors $\ten T$ presented in subsection~\ref{ssec:choice_MT}. 
If $\ten  T = \ten T_B$ is the binary triangle tensor 
in \eqref{eq:TB}, it follows that  
\begin{equation}\label{eq:TA}
    \sum_{k=1}^n (\ten T_B)_{ijk} 
= (A\circ A^2)_{ij}
\end{equation}
and hence 
$$
(\ten T_B)_1(\xvec) = (A\circ A^2)\xvec.
$$
Overall, the map $\map$ then acts on a vector $\xvec$ as follows
$$
\map(\xvec) = \alpha M\xvec + (\ten T_B)_1(\xvec) = \Big( \alpha A + (1-\alpha) (A\circ A^2) \Big) \xvec ,
$$
and so the solution to the constrained eigenvector problem \eqref{eq:eig_gen} is the 
Perron--Frobenius  eigenvector of the matrix $\alpha A + (1-\alpha) (A\circ A^2)$. 
This has a flavour of the work in 
\cite{benson2016higher}, 
where the use of $A\circ A^2$ is advocated as a means to 
incorporate motif counts involving second order structure.
Other choices of the tensor $\ten T$ yield different eigenproblems. 
For example, when $\ten T = \ten T_C$ in (\ref{eq:Tc}) we have 
$$
\sum_{k=1}^n (\ten T_C)_{ijk} 
= 
\begin{cases}
\frac{(A\circ A^2)_{ij}}{d_i(d_i-1)} & \text{if } d_i\geq 2 \\
0 & \text{otherwise}
\end{cases}
$$
and hence \eqref{eq:map} becomes
\[
\map(\xvec) = \alpha M\xvec + (1-\alpha) (\ten T_C)_1(\xvec)  = \Big( \alpha A + (1-\alpha)(D^2-D)^\dagger (A\circ A^2) \Big) \xvec,
\]
where $^\dagger$ denotes the Moore-Penrose pseudo-inverse. 
If we let $\ten T  = \ten T_L$, as defined in \eqref{eq:Tl},  we obtain 
\begin{equation}\label{eq:tmp}
    \sum_{k=1}^n (\ten T_L)_{ijk} 
= \begin{cases}
\frac{(A\circ A^2)_{ij}}{w(i)} & \text{if } d_j\geq 2 \\
0 & \text{otherwise.}
\end{cases} \, 
\end{equation}
Note that in formula \eqref{eq:tmp} 
it is possible to have $\sum_{k=1}^n (\ten T_L)_{ijk}=0$ even in the case where $d_j\geq 2$. This is because, as observed in Proposition  \ref{prop:zero_w}, there are cases where $w(i)>0$ but $i$ does not form any triangle and thus $(A\circ A^2)_{ij}=0$, for all $j$ with $d_j\geq 2$.

Using \eqref{eq:tmp} we obtain
$$
(\ten T_L)_1(\xvec) = W^{\dagger}(A\circ A^2)\xvec,
$$
where $W = \text{diag}(w(1),\ldots,w(n))$. 
The eigenvector problem \eqref{eq:eig_gen} then becomes
$$
\map_{1,\alpha}(\xvec) = \alpha M\xvec + (1-\alpha) (\ten T_L)_1(\xvec) = \Big( \alpha A + (1-\alpha) W^\dagger (A\circ A^2) \Big) \xvec = \lambda\, \xvec.
$$

\subsection{Spectral clustering coefficient: $\alpha = 0$ }
\label{ssec:spectral_CC}

While the choice of $\alpha=1$ yields a linear and purely first order map, the case $\alpha=0$ corresponds to a map that only accounts for second order node relations. 
In particular, this map allows us to define 
spectral, and hence 
mutually reinforcing, versions of 
the Watts--Strogatz 
 clustering 
coefficient \eqref{eq:WSCC}
and the 
local closure coefficient \eqref{eq:local_closure}, where the power mean parameter $p$ in (\ref{eq:Tp}) 
controls how the coefficients of neighbouring nodes are combined.
We therefore make the following definition. 
\begin{definition} \label{def:spec_clus}
Let $\ten T\in\Rnnn$ be an entrywise nonnegative cubic tensor associated with the network.
The \emph{spectral clustering coefficient} of node  $i$ is the 
$i$th entry of the vector $\xvec\geq 0$ which solves the eigenvalue problem 
(\ref{eq:eig_gen}) with $\alpha = 0$ in (\ref{eq:map}); that is, 
\begin{equation}\label{eq:spectral_clustering_coeff}
    \ten T_p(\xvec) = \lambda \xvec\, .
\end{equation}
The solution for 
  $\ten T = \ten T_C\in\Rnnn$ in~\eqref{eq:Tc} will be referred to as the 
  \emph{spectral Watts--Strogatz clustering coefficient}, and 
  the solution for 
  $\ten T = \ten T_L\in\Rnnn$ 
  in~\eqref{eq:Tl}
  will be referred to as the \emph{spectral local closure coefficient}.
\end{definition}

We emphasize that, as for standard first order coefficients based on matrix eigenvectors, the spectral clustering coefficient  \eqref{eq:spectral_clustering_coeff} is invariant under node relabeling. Indeed, if $\ten T$ is any tensor associated with the network as in subsection~\ref{ssec:choice_MT} and $\pi:V\to V$ is a relabeling of the nodes, i.e., a permutation, then the tensor associated to the relabeled graph is $\tilde{\ten T}_{ijk}=\ten T_{\pi(i)\pi(j)\pi(k)}$ and thus $\xvec$ solves \eqref{eq:spectral_clustering_coeff} if and only if $\tilde{\ten T}_p(\tilde \xvec) = \lambda \tilde \xvec$, where $\tilde x_{\pi(i)}=x_i$, for all $i \in V$. Of course, the same relabeling invariance property carries over to the general setting $\alpha\neq 0$.

 \begin{figure}[t]
     \centering
     \includegraphics[width=\textwidth]{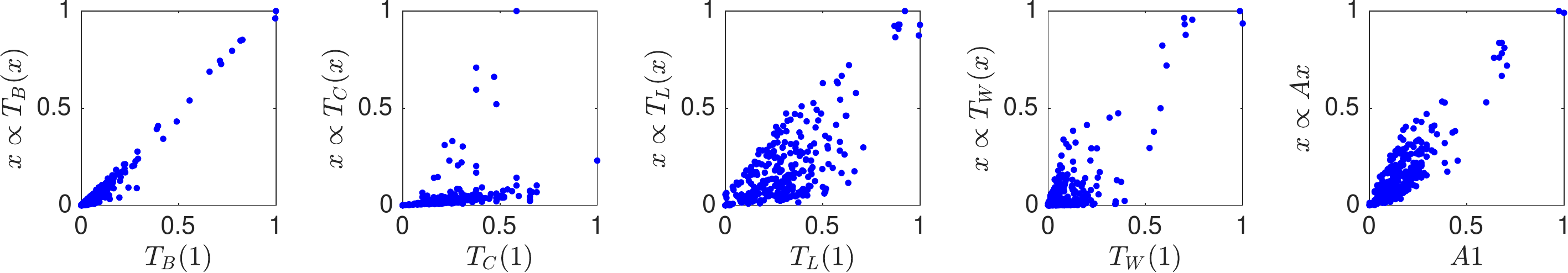}
     \caption{C.\ Elegans neural network data. First four panels: scatter plots showing correlation of static clustering coefficients vs $H$-eigenvector coefficients  for  four choices of the tensor $\ten T$, i.e., solutions to \eqref{eq:spectral_clustering_coeff} for $p=0$. The rightmost panel scatter plots degree centrality $\dvec = A\bone$ vs standard eigenvector centrality.}
     \label{fig:corr_centrality}
 \end{figure}
 
 Note  that if node $i$ does not participate in any triangle, then the summation describing the corresponding entry in $\ten T_p(\xvec)$ is empty, and thus the spectral clustering coefficient for this node is zero, as expected. Moreover, the converse is also true, since $\ten T\geq 0$ and $\xvec\geq 0$.
On the other hand, since the spectral clustering coefficient $\b x$ is defined via an eigenvector equation for $\ten T_p$, it follows that it cannot be unique as it is defined only up to a positive scalar multiple. 
Indeed, we have $\ten T_p(\theta \b x) = \theta \ten T_p(\b x)$ for any $\theta \geq 0$. Hence, when $\ten T = \ten T_C$, unlike the standard Watts-Strogatz clustering coefficient, it is no longer true that a unit spectral clustering coefficient identifies nodes that participate in all possible triangles. 
However, we will see in the next section that once we have a solution $\b x$ of  \eqref{eq:spectral_clustering_coeff} any other solution must be a positive multiple of $\b x$. More precisely, we will show that under standard connectivity assumptions on the network, the spectral clustering coefficient and, more generally, the solution to \eqref{eq:eig_gen} is unique up to  positive scalar multiples. This fosters the analogy with the linear setting \eqref{eq:eig_linear}. 
Therefore, it is meaningful to normalize the solution to \eqref{eq:eig_gen} and compare the size of its components to infer information on the relative importance of nodes within the graph.

 The vector $\ten T_p(\bone)$, which is independent of the choice of $p$, defines a ``static'' counterpart of the spectral clustering coefficient obtained as the Perron--Frobenius eigenvector $\xvec$ of $\ten T_p$. This may be viewed as a second order analogue of the dichotomy between degree centrality and eigenvector centrality, the former being defined as $A\bone$ and the latter as the Perron--Frobenius eigenvector of $A$.  As in the first order case, even though the spectral coefficient $\xvec \propto \ten T_p(\xvec)$ carries global information on the network while the static version $\ten T_p(\bone)$ is highly local, the two measures can be correlated. An example of this phenomenon is shown in Figure \ref{fig:corr_centrality}, which scatter plots $\ten T_p(\bone)$ against $\ten T_0(\xvec)$, for different choices of $\ten T$, on the unweighted version\footnote{We have modified the original weighted network by assigning weight one to every edge.} of the neural network of C. Elegans compiled by Watts and Strogatz in \cite{WS98}, from original experimental data by White et al. \cite{white_celegans}; see Table~\ref{tab:data} for further details of this network.

We also remark that our general definition of spectral clustering coefficient
in Definition~\ref{def:spec_clus}
  includes in the special case 
 $p\to 0$
 the 
 \textit{Perron $H$-eigenvector} of the tensor $\ten T$ \cite{gautier2019unifying}. Indeed, it is easy to observe that the change of variable $\b y^2 = \xvec$ yields
$$
\ten T_0(\xvec)=\lambda \xvec \, \iff\,  \ten T\b y \b y = \lambda \b y^2,
$$
where $\ten T\b y \b y$ is the tensor-vector product $(\ten T\b y \b y)_i = \sum_{jk} \ten T_{ijk}y_jy_k$.  This type of eigenvector has been used in the context of hypergraph centrality; see, e.g., \cite{B19}.

The choice of the tensor $\ten T$ affects the way the 
triangle structure is
incorporated in our measure, as we have previously illustrated in the small toy networks in Figure~\ref{fig:multi_toy_networks}. 
Examples of the differences that one may obtain on real-world networks are shown in  Figures \ref{fig:karate} and \ref{fig:correlation_clustering_coeff}. A description of the datasets used in those figures is provided in Section~\ref{sec:numerical} and in Table~\ref{tab:data}. 
Figure~\ref{fig:karate} displays the {\sc Karate} network and highlights the  ten nodes which score the highest according to the spectral clustering coefficient for different choices of $\ten T$. 
In this experiment we select $p=0$, and thus we are actually computing the Perron $H$-eigenvector of the corresponding tensors. The size of each of the top ten nodes in Figure~\ref{fig:karate} is proportional to their clustering coefficients. In Figure~\ref{fig:correlation_clustering_coeff}, instead, we display how the $H$-eigenvectors corresponding to different triangle tensors correlate with the degree of the nodes for four real-world networks; 
We group nodes by logarithmic binning of their degree and plot the average degree versus the average clustering coefficient in each bin.
As expected, the Watts--Strogatz spectral clustering coefficient may decrease when the degree increases, in contrast with other choices of the triangle tensor. A similar phenomenon is observed, for example, in \cite{yin2019local}. 

\begin{figure}[t]
\centering
\includegraphics[width=\textwidth]{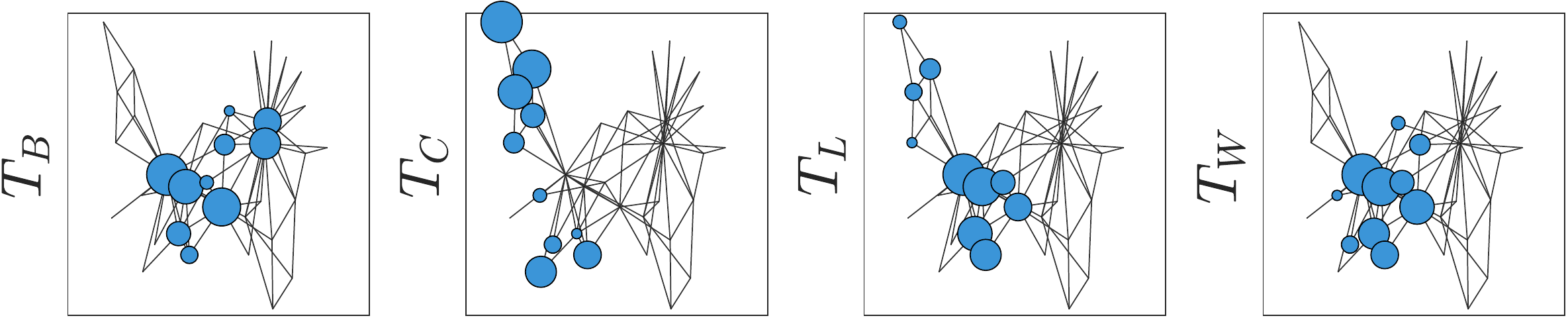}
\caption{Top 10 nodes identified on the \textit{Karate Club network} by different tensor $H$-eigenvector  clustering coefficients, solution to $\ten T_p(\xvec) = \lambda \xvec$, for $p=0$, and the four triangle tensor choices $\ten T \in \{\ten T_B, \ten T_C, \ten T_L, \ten T_W\}$. }\label{fig:karate}
\end{figure}
\begin{figure}[t]
    \centering
    \includegraphics[width=\textwidth]{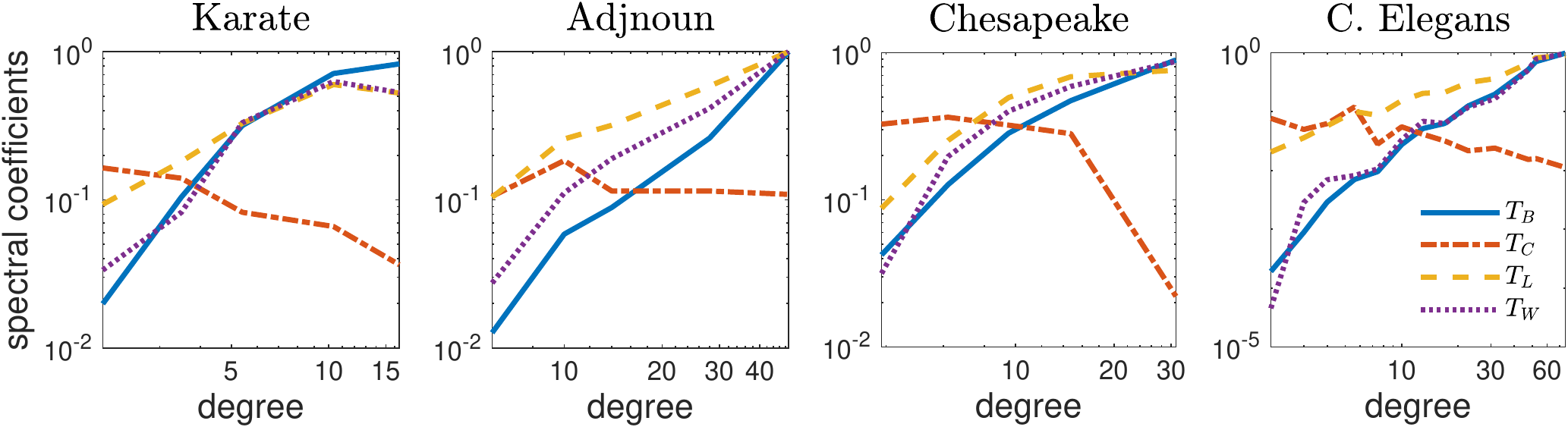}
    \caption{Correlation of different tensor $H$-eigenvector  clustering coefficients with node degree on four networks. We group nodes by logarithmic binning of their degree and plot the average degree versus the average clustering coefficient in each bin.}
    \label{fig:correlation_clustering_coeff}
\end{figure}

In the next section we discuss existence and uniqueness, up to scalar multiples, of a solution to \eqref{eq:eig_gen}. We also describe a power-iteration algorithm for its computation.

\section{Existence, uniqueness, maximality and computation}
\label{sec:theory}

\def\mx{B}

For reasons of clarity and utility, the definitions in Section~\ref{sec:model} 
were made under
the assumption that the original graph is undirected.
Second order features can, of course, be incorporated in the directed case. But the 
range of possibilities to be considered (for example, 
accounting for each type of directed triangle) is much greater and   
the interpretation of the resulting measures becomes less clear cut.
 However, just as in the standard matrix setting,
 in terms of studying existence, uniqueness and computational issues, very little is lost 
 by moving to the unsymmetric case.
 Hence, Definition~\ref{def:A_M} and Theorem~\ref{thm:theory} below are stated for general $M$ and $\ten T$.
 In Lemma~\ref{lem:graph_of_M} we then clarify that the results apply to the measures
 introduced in Section~\ref{sec:model}.
 
We begin by discussing the linear case
where 
$\alpha=1$ or $p=1$, so that 
the nonnegative operator $\map:\RR^n\to\RR^n$ is an entrywise nonnegative matrix $\mx$.
Here, results from Perron--Frobenius theory provide conditions on $\map$ that guarantee existence of a solution to \eqref{eq:eig_gen} and computability of this solution via the 
 classical power method. 
These conditions are typically based on structural properties of  $\map$ and of the associated graph. 
We review below some of the best known and most useful results from this theory. 

First, given the entrywise nonnegative matrix $\mx\in \mathbb R^{n\times n}$, let $G_\mx$ be the adjacency graph of $\mx$, with nodes in $\{1,\dots,n\}$ and such that the edge $i\to j$ exists in $G_\mx$ if and only if $\mx_{ij}>0$. 
Now, recall that  a graph is said to be \emph{aperiodic} if the greatest common divisor of the lengths of all cycles in the graph is one. 
Also, the matrix $\mx$ is \emph{primitive} if and only if there exists an integer $k\geq 1$ such that $\mx^k>0$, and, moreover,  
 $\mx\geq 0$ is {\rm primitive} if and only if $G_\mx$ is aperiodic. 

It is well known that when $G_\mx$ is strongly connected, then there exists a unique (up to  multiples) eigenvector of $\mx$, and such vector is entrywise positive.
Moreover, this eigenvector is maximal, since the corresponding eigenvalue is the spectral radius of $\mx$ and, if $G_\mx$ is aperiodic, the power method iteration $\xvec_{k+1} = \mx\xvec_k / \|\mx\xvec_k\|$ converges to it for any starting vector $\xvec_0 \in \RR^n$.

In the general case, 
we will appeal to nonlinear Perron--Frobenius theory to show that 
 the properties of existence, uniqueness and maximality of the solution to \eqref{eq:eig_gen} carry over to the general nonlinear setting almost unchanged, and to 
 show that an efficient iteration can be used to compute this solution.
We first note that for any $\alpha\in [0,1]$, any $p\in \RR$ and any $\theta > 0$ we have
\[
\map(\theta \b x) = \alpha M (\theta \b x) + (1-\alpha)\b T_p(\theta \b x) = \theta \map (\b x)\, ,
\]
thus if $\b x\geq 0$ solves \eqref{eq:eig_gen}, then any positive multiple of $\b x$ does as well. Therefore, as for the linear case, uniqueness can only be defined up to scalar multiples. 
We continue by introducing the graph of $\map$. 

\begin{definition}\label{def:A_M}
Given a matrix $M\in\Rnn$ and a cubic tensor $\ten T\in\Rnnn$, both assumed to be nonnegative, we define the {\rm adjacency graph $G_\map$ of} $\map$ in \eqref{eq:map} as the pair $G_\map=(V,E_\map)$ where $V=\{1,\dots,n\}$ and, for all $i,j\in V$,  $(i,j)\in E_\map$ if and only if $(A_\map)_{ij}=1$, where $A_\map$ is the adjacency matrix entrywise defined as
\begin{equation*}
(A_\map)_{ij}=\begin{cases}1 & \text{if }\,\alpha M_{ij} + (1-\alpha)\sum_{k=1}^n(\ten T_{ijk}+\ten T_{ikj})>0\\
0 & \text{otherwise}
\end{cases}
\end{equation*}
\end{definition}

We now state and prove our main theorem.

\begin{theorem}\label{thm:theory}
Given the nonnegative matrix $M\in \RR^{n\times n}$ and the nonnegative tensor $\ten T\in \RR^{n\times n\times n}$, let  $\map$ be defined as in \eqref{eq:map} and let $G_\map$ be its adjacency graph, as in Definition \ref{def:A_M}. If $G_\map$ is strongly connected, then
\begin{enumerate}
    \item  There exists a unique (up to multiples) positive eigenvector of $\map$, i.e. a unique positive solution of \eqref{eq:eig_gen}.
    \item The positive eigenvector of $\map$ is maximal, i.e., its eigenvalue is $\rho(\map) = \max \{|\lambda| : \map(\b x) = \lambda \b x\}$.
    \item If $\b x$ is any nonnegative eigenvector $\map(\b x)=\lambda \b x$ with some zero entry, then $\lambda<\rho(\map)$.
    \end{enumerate}
If moreover $G_\map$ is aperiodic, then
\begin{enumerate}
    \item[(iv)] For any starting point $\b x_0> 0$, the nonlinear power method
    $$
    \left\{\begin{array}{l}\b y_{k+1} = \alpha M\b x_k + (1-\alpha) \ten T_p(\b x_k) \\
    \b x_{k+1} = \b y_{k+1}/\|\b y_{k+1}\|
    \end{array}\right.$$
    converges to the positive eigenvector of $\map$. Moreover, for all $k=0,1,2,...$ it holds
    \begin{equation}\label{eq:CW}
            \min_{i=1,\dots,n}\frac{(\b y_k)_i}{(\b x_k)_i} \leq  \min_{i=1,\dots,n}\frac{(\b y_{k+1})_i}{(\b x_{k+1})_i} \leq \rho(\map) \leq \max_{i=1,\dots,n}\frac{(\b y_{k+1})_i}{(\b x_{k+1})_i} \leq \max_{i=1,\dots,n}\frac{(\b y_k)_i}{(\b x_k)_i}
    \end{equation}
    with both the left and the right hand side sequences converging to $\rho(\map)$ as $k\to \infty$.
\end{enumerate}
\end{theorem}
\begin{proof}
The proof combines several results from nonlinear Perron--Frobenius theory.

First, note that $\map$ is homogeneous of degree one and order preserving. Indeed, if $\b x\geq \b y \geq 0$ entrywise, then it is easy to verify that 
$$
\map(\b x) = \alpha M\b x + (1-\alpha)\ten T_p(\b x) \geq \alpha M\b y + (1-\alpha)\ten T_p(\b y) = \map(\b y) \geq 0 \, .
$$
It follows that 
$\map$ has at least one entrywise nonnegative eigenvector that corresponds to the eigenvalue $\lambda = \rho(\map)$ (see, e.g., \cite[Theorem 5.4.1]{lemmens_nussbaum}).

Next, recall that $\bone_j$ denotes the $j$th vector of the canonical basis of $\Rn$. 
Now let $\b y_j(\beta) =\bone + (\beta -1)\bone_{j}$ be the vector whose $j$th component is the variable $\beta\in \RR$ while all the other entries are equal to one. 
Thus note that if  ${A_{\map}}_{ij}=1$, then $\lim_{\beta\to\infty}\map(\b y_j(\beta)_i)=\infty$. 
Since $G_\map$ is strongly connected,  \cite[Theorem 1]{gaubert} implies that  $\map$ has at least one entrywise  positive eigenvector $\b u>0$ such that $\map(\b u) = \widetilde \lambda \b u$, with $\widetilde \lambda  >0$.

Third, we show uniqueness and maximality. Note that for any positive vector $\b y>0$ and any $p\geq 0$ we have that if $G_\map$ is strongly connected then the Jacobian matrix of $\map$ evaluated at $\b y$ is irreducible. In fact 
$$
\frac{\partial}{\partial x_j} \map(\b y)_i = \alpha M_{ij} + (1-\alpha)y_j^{p-1} \sum_k (\ten T_{ijk} + \ten T_{ikj}) \mu_p(y_j,y_k)^{1-p} \, .
$$
Therefore, \cite[Theorem 6.4.6]{lemmens_nussbaum} 
implies that $\b u$ is the unique positive eigenvector of $\map$. Moreover, \cite[Theorem 6.1.7]{lemmens_nussbaum} implies that for any other nonnegative eigenvector $\b x\geq 0$ with $\map(\b x) = \lambda \b x$ we have $\lambda < \rho(\map)$. As there exists at least one nonnegative eigenvector corresponding to the spectral radius, that must be $\b u$ and we deduce that $\widetilde \lambda  = \rho(\map)$.  

This proves points $(i)-(iii)$. For point $(iv)$, we note that if $G_\map$ is aperiodic then $A_\map$ is primitive and this implies that the Jacobian matrix of $\map$ evaluated at $\b u>0$ is primitive as well. Thus Theorem 6.5.6 and Lemma 6.5.7 of \cite{lemmens_nussbaum} 
imply that the normalized iterates of the homogeneous and order preserving map $\map$ converge to $\b u$. Finally, \cite[Theorem 7.1]{multiPF} proves the sequence of inequalities in \eqref{eq:CW} and the convergence of both the sequences 
$$
\alpha_k =  \min_{i=1,\dots,n}\frac{\map(\b x_k)_i}{(\b x_k)_i} \quad \text{and} \quad  \beta_k = \max_{i=1,\dots,n}\frac{\map(\b x_k)_i}{(\b x_k)_i} 
$$
towards the same limit; $\alpha_k$ and $\beta_k$ tend to 
$\rho(\map)$ as $k \to \infty$.
\end{proof}

We emphasize that because the mapping  $\map$ of Theorem \ref{thm:theory} is defined for an arbitrary nonnegative matrix $M$ and nonnegative tensor $\ten T$, the graph $G_\map$ in that theorem may be directed. 
The next lemma shows that 
when  $M$ and $\ten T$ are defined 
as  in 
subsection~\ref{ssec:choice_MT} 
the graph $G_\map$ coincides with the underlying network. 
Thus, for the undirected case and with any of the choices in 
subsection~\ref{ssec:choice_MT},  
Theorem~\ref{thm:theory} applies   
whenever the original graph is connected.

\begin{lemma}\label{lem:graph_of_M}
Let $\alpha\neq 0$ and $M$ and $\ten T$ be defined according to any of the choices in  subsection~\ref{ssec:choice_MT}. Then $M$ and $A_\map$ have the same sparsity pattern; that is,  $M_{ij}>0$ if and only if $(A_\map)_{ij}=1$. 
\end{lemma}
\begin{proof}
If $(i,j)\in E$ is an edge in the graph associated with $M$, i.e.\ $M_{ij}>0$,  then clearly $(A_\map)_{ij}=1$ as the tensor $\ten T$ has nonnegative entries. 
If $(i,j)\not\in E$, then from the possible definitions of the tensor $\ten T$ listed in Subsection~\ref{ssec:choice_MT} it follows that $\ten T_{ijk} = \ten T_{ikj} = 0$, for all $k$. Thus $(A_\map)_{ij} = M_{ij} = 0$.
Vice versa, if $(A_\map)_{ij} = 0$, then $\alpha M_{ij} + (1-\alpha)\sum_{k=1}^n(\ten T_{ijk}+\ten T_{ikj}) = 0$. Since we are summing two nonnegative terms, it follows that 
both are zero and, in particular, $M_{ij} = 0$. If $(A_\map)_{ij}=1$, on the other hand, this implies $\alpha M_{ij} + (1-\alpha)\sum_{k=1}^n(\ten T_{ijk}+\ten T_{ikj})>0$ and hence at least one of the two terms has to be positive; however, from the possible definitions of $\ten T$ it is clear that $\ten T_{ijk}$ and $\ten T_{ikj}$ cannot be nonzero unless $(i,j)\in E$, i.e., unless $M_{ij}>0$. 
\end{proof}

\section{Example network with theoretical comparison}
\label{sec:as}
In this section we describe theoretical results on the higher order centrality
measures. 
Our overall aim is to confirm that the incorporation of second order information can 
make a qualitative difference to the rankings. 
We work with networks of the form represented in Figure~\ref{fig:asymptotic_network}. These have three different types of  nodes: i) node $1$, the centre of the wheel, that has degree $m$ and connects to $m$ nodes of the second type, ii) $m$ nodes attached to node $1$ and interconnected via a cycle to each other. Each type (ii) node also connects to $k$ nodes of the third type, and iii) $mk$ leaf nodes attached in sets of $k$ to the $m$ nodes of type (ii). 
We will use
node $2$ to represent the nodes of type (ii) and 
node $m+2$ to represent the nodes of type (iii).

The network is designed so that node $1$ is connected to important nodes and is also
involved in many triangles. Node 2, by contrast, is only involved in two triangles
and has connections to the less important leaf nodes.  
If we keep $m$ fixed and increase the number of leaf nodes, $k$, 
then
eventually we would expect the centrality of node $2$ to overtake that of node $1$.
We will show that this changeover happens for a larger value of $k$ when we incorporate second order information.
More precisely,  
we set $p=1$ and 
show that node $1$ is identified by the higher-order measure as being more central than node $2$ for larger values of $k$ when compared with standard eigenvector centrality.

\begin{figure}
    \centering 
    \begin{tikzpicture}[scale=.9]
    \draw (0,0) node[scale=.7,circle,draw=black,thick](1){1};
    \draw (360/5: 1.5cm) node[scale=.8,circle,draw=black,thick](2){2};
    \draw (2*360/5: 1.5cm) node[scale=.8,circle,draw=black,thick](3){3};
    \draw (3*360/5: 1.5cm) node[scale=.8,circle,draw=black,thick](4){4};
    \draw (4*360/5: 1.5cm) node[scale=.8,circle,draw=black,thick](5){5};
    \draw (5*360/5: 1.5cm) node[scale=.5,circle,draw=black,thick](6){$m+1$};
    \draw (360/20: 2.5cm) node[scale=.7,circle,draw=black,thick](6c){};
    \draw (3*360/20: 2.5cm) node[scale=.7,circle,draw=black,thick](2a){};
    \draw (4*360/20: 2.5cm) node[scale=.7,circle,draw=black,thick](2b){};
    \draw (5*360/20: 2.5cm) node[scale=.7,circle,draw=black,thick](2c){};
    \draw (7*360/20: 2.5cm) node[scale=.7,circle,draw=black,thick](3a){};
    \draw (8*360/20: 2.5cm) node[scale=.7,circle,draw=black,thick](3b){};
    \draw (9*360/20: 2.5cm) node[scale=.7,circle,draw=black,thick](3c){};
    \draw (11*360/20: 2.5cm) node[scale=.7,circle,draw=black,thick](4a){};
    \draw (12*360/20: 2.5cm) node[scale=.7,circle,draw=black,thick](4b){};
    \draw (13*360/20: 2.5cm) node[scale=.7,circle,draw=black,thick](4c){};
    \draw (15*360/20: 2.5cm) node[scale=.7,circle,draw=black,thick](5a){};
    \draw (16*360/20: 2.5cm) node[scale=.7,circle,draw=black,thick](5b){};
    \draw (17*360/20: 2.5cm) node[scale=.7,circle,draw=black,thick](5c){};
    \draw (19*360/20: 2.5cm) node[scale=.7,circle,draw=black,thick](6a){};
    \draw (20*360/20: 2.5cm) node[scale=.7,circle,draw=black,thick](6b){};
    \path[-, thick] 
    (1)edge[] node[]{} (2) (1)edge[] node[]{} (3) (1)edge[] node[]{} (4) (1)edge[] node[]{} (5) (1)edge[] node[]{} (6)
    (2)edge[] node[]{} (3)  (2)edge[] node[]{} (6) (3)edge[] node[]{}(4)  (4)edge[] node[]{} (5);
    \path[-, thick]
    (2)edge[] node[]{} (2a) 
    (2)edge[] node[]{} (2b) 
    (2)edge[] node[]{} (2c) 
    (3)edge[] node[]{} (3a)
    (3)edge[] node[]{} (3b) 
    (3)edge[] node[]{} (3c) 
    (4)edge[] node[]{} (4a)
    (4)edge[] node[]{} (4b) 
    (4)edge[] node[]{} (4c) 
    (5)edge[] node[]{} (5a)
    (5)edge[] node[]{} (5b) 
    (5)edge[] node[]{} (5c) 
    (6)edge[] node[]{} (6a)
    (6)edge[] node[]{} (6b) 
    (6)edge[] node[]{} (6c); 
    \path[dotted,thick] 
    (5)edge[] node[]{} (6)
    (2c)edge[] node[]{} (2b)
    (3c)edge[] node[]{} (3b) 
    (4c)edge[] node[]{} (4b) 
    (5c)edge[] node[]{} (5b)
    (6c)edge[] node[]{} (6b)
    ; 
    \end{tikzpicture}
    \caption{Representation of the network used in Section~\ref{sec:as}. The network is a modified wheel graph where each of the $m$ nodes on the cycle are connected to $k$ leaves.}
    \label{fig:asymptotic_network}
\end{figure}
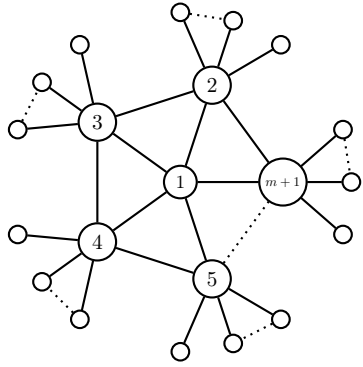

With this labeling of the nodes, the adjacency matrix $A\in\R^{n\times n}$ of the network has the form
\[
A = \left[
\begin{array}{c|c|ccc}
     0 & \bone_m^T & 0 & \cdots & 0 \\
     \hline
     \bone_m & C & & I_m\otimes \bone_k^T & \\
     \hline
     0 & & & & \\
     \vdots & I_m\otimes\bone_k & & & \\
     0 & & & & \\
     \end{array}\right], \qquad C = \left[
\begin{array}{cccc}
    0 & 1      &        & 1 \\
    1 & \ddots & \ddots &   \\
      &  \ddots& \ddots & 1 \\
    1 &        & 1      & 0  
\end{array}
\right] \in \R^{m\times m}\, .
\]
The 
eigenvector $\vvec = [x\;\; y\bone_m^T\;\; z\bone_{mk}^T]^T$ associated to the leading eigenvalue $\lambda = 1+\sqrt{1+m+k}$ of $A$ is such that 
$\lambda x = m\,y$
and it can be verified that 
\[ x > y \quad \text{ if and only if } \quad k < m(m-3).
\]

We now move on to the higher order setting. We begin by specifying the entries of the binary triangle tensor $\ten T_B = (\ten T_B)_{ijk}$ defined in \eqref{eq:TB}. 
It is clear that $(\ten T_B)_{ijk} = 0$ for all $i=m+2,\ldots,mk+m+1$. 
Moreover,
\[
(\ten T_B)_{1jk} = \begin{cases}
1 & \text{if  $j, k = 2,\ldots,m+1$ are such that $(j,k)\in E$}\\
0 & \text{otherwise},
\end{cases}
\]
and for $i = 2,\ldots,m+1$
\[
(\ten T_B)_{ijk} = 
\begin{cases}
1 & \text{if either $j=1$ and $(i,k)\in E$ or $k=1$ and $(i,j)\in E$} \\
0 & \text{otherwise.}
\end{cases}
\]
Using~\eqref{eq:Tp} it follows that
\[
((\ten T_B)_p(\vvec))_1  = 2my,\qquad ((\ten T_B)_p(\vvec))_2 = 4\mu_p(x,y),\qquad ((\ten T_B)_p(\vvec))_{m+2} = 0,
\]
where $\vvec = [x\;\; y\bone_m^T\;\; z\bone_{mk}^T]^T$ as before. 
Overall we thus have that equation \eqref{eq:eig_gen} rewrites as
\[
\left\{
\begin{array}{l}
    \lambda x = (2-\alpha) my \\
    \lambda y = \alpha (x + 2y + kz) + 4(1-\alpha)\mu_p(x,y)\\
    \lambda z = \alpha y. 
\end{array}
\right.
\]
For $p = 1$ and $\alpha\in(0,1]$, this system yields 
\[
x>y \quad \text{ if and only if } \quad k < \frac{(2-\alpha)}{\alpha^2}\left((2-\alpha)m^2  + (\alpha -4)m\right).
\]
The areas for which $x>y$ in the two settings (standard eigenvector centrality $\alpha = 1$ and higher order centrality $\alpha = 0.2,0.5$) are shaded in Figure~\ref{fig:shadeT} (left). 
It is readily seen that even for small values of $m$, $k$ needs to become very large (when compared to $m$) in order for the centrality of nodes $i = 2,\ldots, m+1$ to become larger than that of node $1$ when higher-order information is taken into account. 
In the standard eigenvector centrality setting we observe a very different behaviour (see Figure~\ref{fig:shadeT}, left, $\alpha = 1$). 

\begin{figure}
    \centering
    \includegraphics[width=\textwidth]{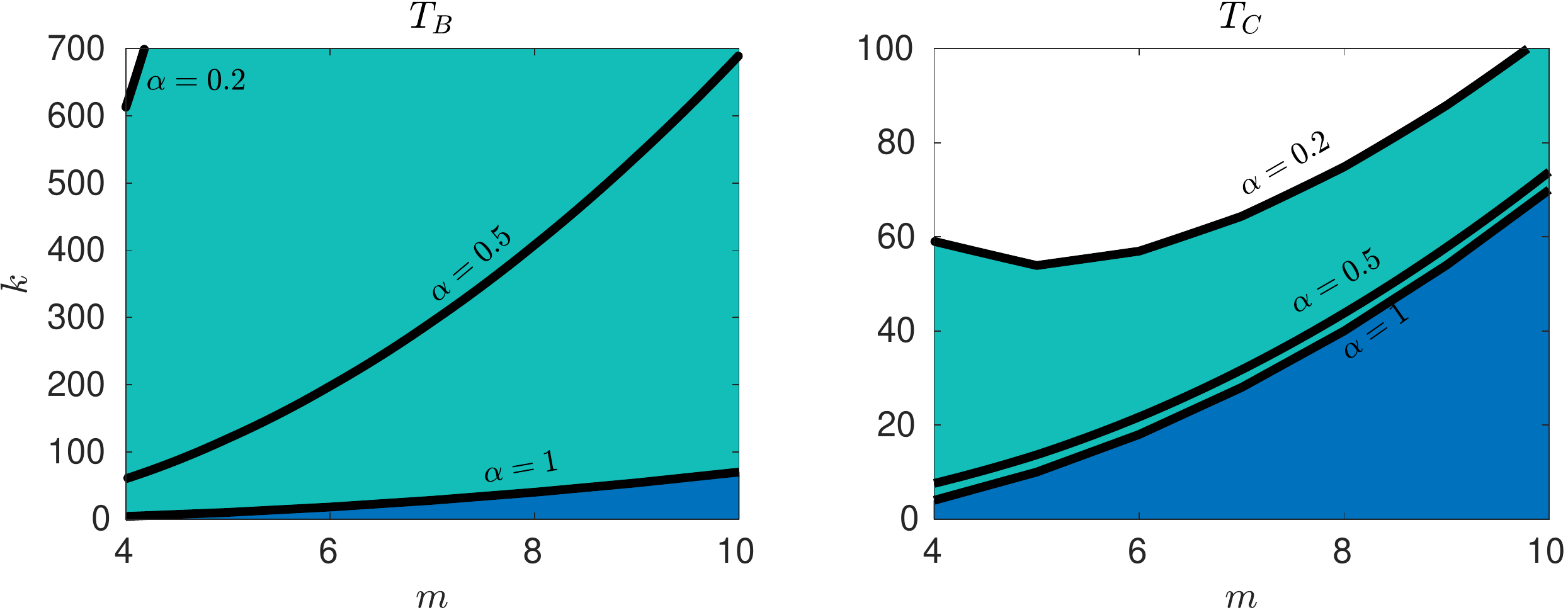}
    \caption{Values of $m$ and $k$ for which $x>y$ (shaded) for different values of $\alpha$, $p=1$ and tensors $\ten T_B$ (left) and $\ten T_C$ (right).}
    \label{fig:shadeT}
\end{figure}

In Figure~\ref{fig:shadeT} (right) we display the areas for which $x>y$ for different values of $\alpha$ when $\ten T_C\in\Rnnn$ is used in \eqref{eq:eig_gen}. 
Indeed, specializing the definition in~\eqref{eq:Tc} to this example, it is easy to see that 
\[
((\ten T_C)_p(\vvec))_1 = \frac{2y}{m-1}, \qquad  ((\ten T_C)_p(\vvec))_2 = \frac{4\mu_p(x,y)}{(k+3)(k+2)}, \qquad ((\ten T_C)_p(\vvec))_{m+2} = 0, 
\]
and therefore the solution to \eqref{eq:eig_gen} must satisfy 
\[
\left\{
\begin{array}{l}
    \lambda x = \left(\alpha m + \frac{2(1-\alpha)}{m-1}\right)y \\
    \lambda y = \alpha (x + 2y + kz) + \frac{4(1-\alpha)}{(k+3)(k+2)}\mu_p(x,y)\\
    \lambda z = \alpha y. 
\end{array}
\right.
\]
After some algebraic manipulation, we obtain that 
\[
x > y \quad \text{ if and only if } \quad \alpha m + \frac{2(1-\alpha)}{m-1}>\lambda. 
\]
If we now let $p = 1$, it is easy to show that $\lambda$ satisfies 
\begin{equation}\label{eq:lambda_toy}
\lambda^2 - (2\alpha + c_1)\lambda - (\alpha +c_1)(\alpha m + c_2) - k\alpha^2 = 0
\end{equation}
with $c_1 = \frac{2(1-\alpha)}{(k+3)(k+2)}$ and $c_2 = \frac{2(1-\alpha)}{m-1}$. 

\begin{remark}
Similarly, if the local closure triangle tensor $\ten T_L$ is used in the computation, we observe that $x>y$ if and only if $\alpha m + \frac{2(1-\alpha)}{k+2}>\lambda$ where now $\lambda$ satisfies \eqref{eq:lambda_toy} for $c_1 = \frac{2(1-\alpha)}{m+2k+3}$ and $c_2 = \frac{2(1-\alpha)}{k+2}$. There seems to be no appreciable difference between the profiles for $\alpha = 0.2, 0.5, 1$, and hence they are not displayed here. 
\end{remark}

\section{Applications and numerical results}\label{sec:numerical}

\subsection{Centrality measures}\label{ssec:num_cen}
In the previous subsection we observed that $\alpha$ 
may have a significant effect on the node rankings. 
Results were shown for $\ten T_B$ and $\ten T_C$, $p=1$ and $\alpha = 0.2, 0.5, 1$. 
In this subsection we test the role of $\alpha$ for real network data. We use $\alpha = 0.5$ and $\alpha = 1$ (corresponding to eigenvector centrality) and $p=0$ in \eqref{eq:eig_gen}, and
combine the adjacency matrix $A$ and the binary tensor $\ten T_B$. 

Our tests were performed on four real-world networks which are often used as benchmarks in the graph clustering and community detection communities, and are publicly available at \cite{SuiteSparse}. 
The {\sc Karate} network is a social network representing  friendships between the 34      
members of a karate club at a US university~\cite{karate}. The network {\sc C. Elegans} is a neural network. We use here an undirected and unweighted version of the neural network of C. Elegans compiled by Watts and Strogatz in \cite{WS98}, from original experimental data by White et al. \cite{white_celegans}. 
The network {\sc Adjnoun} is based on common adjective and noun    adjacencies in the novel ``David Copperfield'' by Charles Dickens~\cite{newman2006finding}.  
{\sc Chesapeake} represents the interaction network of the Chesapeake Bay ecosystem. 
Here, nodes represent species or suitably defined functional groups and links create the food web \cite{chesapeake}. 
In Table~\ref{tab:data} we report the number of nodes $n$, (undirected) edges $m$ and triangles $\triangle = \text{trace}(A^3)/6$ for the four networks. 
We further display the global clustering coefficient $\widehat{C}$, the average clustering coefficient $\overline{c}$, and the average spectral clustering coefficient $\overline{\xvec}_C$, as well as the average local closure coefficient $\overline{\wvec}$ \cite{yin2019local} and its spectral counterpart $\overline{\xvec}_L$; see Defintion~\ref{def:spec_clus}. 

\begin{table}[]
    \centering
    \begin{tabular}{c|*{8}c}
    \hline
    Name & $n$ & $m$ & $\triangle$ & $\widehat{C}$ & $\overline{\cvec}$ & $\overline{\xvec}_C$ & $\overline{\wvec}$ & $\overline{\xvec}_L$ \\
    \hline
    \hline
        \sc{Karate} & 34 & 78 & 45  & 0.26 & 0.57 & 0.12 & 0.22 & 0.23\\
        \sc{Chesapeake} & 39 & 170 & 194 & 0.28 & 0.45 &  0.41 & 0.25 & 0.38\\
        \sc{Adjnoun} & 112 & 425 & 284 & 0.16 & 0.17 &  0.18 & 0.09 & 0.18\\
        \sc{C. Elegans} & 277 & 1918 & 2699 & 0.19 & 0.28 & 0.05 &  0.15 & 0.20\\
        \hline
    \end{tabular}
    \caption{Description of the dataset: $n$ is the number of nodes, $m$ is the number of edges, $\triangle$ is the number of triangles, $\widehat{C}$ is the global clustering coefficient of the network, $\overline{\cvec}$ is the average clustering coefficient, $\overline{\xvec}_C$ is the average spectral clustering coefficient, $\overline{\wvec}$ is the average local closure coefficient, and $\overline{\xvec}_L$ is the average spectral local closure coefficient.}
    \label{tab:data}
\end{table}

 \begin{figure}[!t]
     \centering
     \includegraphics[width=\textwidth]{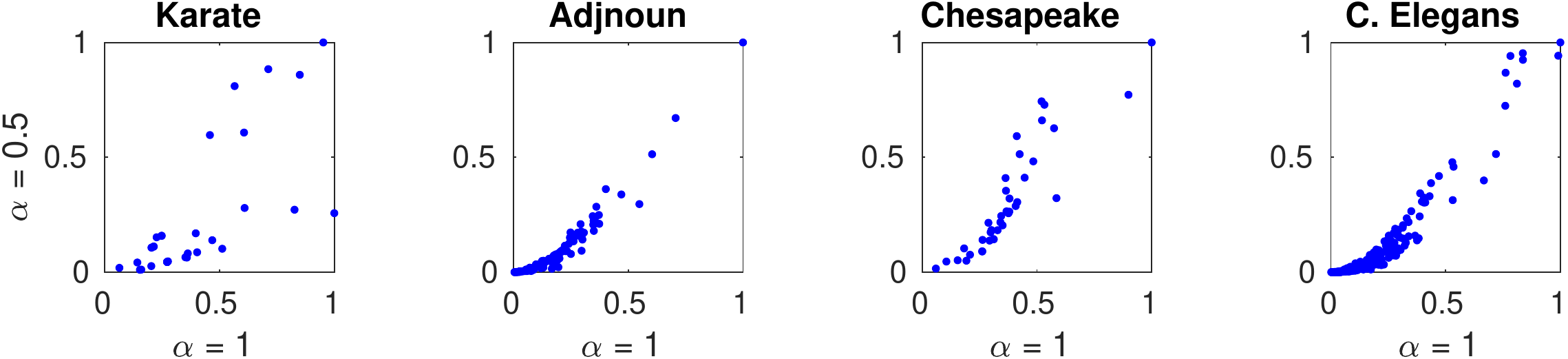} 
     \caption{Scatter plot of the solution to~\eqref{eq:eig_gen} with $M=A$ and $\ten T = \ten T_B$. The plot shows the solution for $\alpha = 0.5$ and $p=0$ versus standard eigenvector centrality, i.e., \eqref{eq:eig_gen} for $\alpha = 1$.}
     \label{fig:scatter}
 \end{figure}
 
Figure~\ref{fig:scatter} scatter plots the newly introduced measure against eigenvector centrality for the four different networks. The centrality vectors are normalized with the infinity norm. 
For the network {\sc Karate} we see very poor correlation between the two measures. 
Stronger correlation is displayed for the other networks, but it is still to be noted that the top ranked nodes (corresponding to the nodes with largest centrality scores) differ for the two measures in all but one network, namely {\sc Adjnoun}. 
Hence, using second order information 
can alter our conclusions about which nodes are the most central.

\subsection{Link Prediction}\label{ssec:num_link}
Link prediction is a fundamental task in network analysis: given a network $G_0 = (V,E_0)$,
we must identify edges that are not in $E_0$ but should be there. 
This problem typically arises in two settings:
(a) in a dynamic network where new connections appear over time, and 
(b) in a noisily observed network, where it is suspected that edges are missing  \cite{clauset2008hierarchical,liben2007link,lu2011link}.    

For convenience, let us assume that $E_0$ is the set of edges that we observe and that $E_1$ with $E_1\cap E_0=\emptyset$ is the set of edges that should be predicted, i.e., those that will appear in an evolving network or that are missing in a noisy graph. A standard approach for link prediction is to create a {\it similarity matrix} $S$,  whose entries $S_{ij}$ quantify the probability that $(i,j)\in E_1$. It is worth pointing out that since $E_0\cap E_1 = \emptyset$, then the nonzero pattern of $S$ will be complementary to that of the adjacency matrix of $G_0$. 
Over the years, several similarity measures have been proposed in order to quantify which nodes are most likely to link to a given node $i$ \cite{martinez2017survey}. While classical methods usually exploit the first order structure of connections around $i$, there is a growing interest in second order methods that take into account, for example, triangles. 

In this context, we propose a new similarity measure based on $\map$ and its Perron eigenvector.
This measure is a generalization of a well-known technique known as {\it seeded} (or \textit{rooted}) \textit{PageRank} \cite{gleich2016seeded,jeh2003scaling}, which we
now describe. 
Given a seed node $\ell \in V$ and a teleportation coefficient $0\leq c < 1$, let $\xvec^{(\ell)}$ be the limit of the evolutionary process
\begin{equation}\label{eq:pr_link}
    \xvec_{k+1} = c P\xvec_{k} + (1-c)\bone_\ell, \qquad k=0,1,2,\dots
\end{equation}
where $P$ is the random walk matrix $P=AD^{-1}$. As $0 \leq  c <  1$,  it is easy to show that the limit exists and that it coincides with 
the solution to the linear system
\begin{equation}\label{eq:lsys}
    (I-cP)\xvec^{(\ell)} = (1-c)\bone_\ell\, .
\end{equation}
The seeded PageRank similarity matrix $S_{PR}$ is then entrywise defined by 
$$
(S_{PR})_{ij} = (\xvec^{(i)})_j + (\xvec^{(j)})_i\, .
$$
The idea behind \eqref{eq:pr_link} is that the sequence $\xvec_k$ is capturing the way a unit mass centered in $\ell$ (the {\it seed} or {\it root} of the process), and represented in the model by $\bone_\ell$,  propagates throughout the network following the diffusion rule described by $P$. 
This diffusion map is a first order random walk on the graph. 

In order to propose a new, second order, similarity measure, we replace this first order map with the second order diffusion described by $\map = \alpha M+ (1-\alpha)\ten T_p$ and we consider the associated diffusion process. 
To this end, we begin by observing that, independently of the choice of the starting point $\xvec_0$ in \eqref{eq:pr_link}, the first order diffusion process will always converge to $\xvec^{(\ell)}$ that satisfies $\|\xvec^{(\ell)}\|_1=1$. Indeed, \eqref{eq:lsys} yields
$$
\|\xvec^{(\ell)}\|_1 = (1-c)\|\sum_{k\geq 0}c^k P^k \bone_\ell\|_1 = (1-c)\sum_{k\geq 0}c^k=1\, .
$$
As a consequence, the limit of the sequence \eqref{eq:pr_link} coincides with the limit of the normalized iterates $\hat \xvec_{k+1} = c P\xvec_{k} + (1-c)\bone_\ell$, with $\xvec_{k+1} = \hat \xvec_{k+1}/\|\hat \xvec_{k+1}\|_1$. 
On the other hand, when the linear process $P$ is replaced by the nonlinear map $\map$, the unnormalized sequence may not converge. 
We thus need to impose normalization of the vectors in our dynamical process defined in terms of $\map$ and seeded in the node $\ell$: 
\begin{align}\label{eq:us_link}
\begin{aligned}
     &\hat{\yvec}_{k+1} = c \map(\yvec_k) + (1-c)\bone_\ell\qquad k=0,1,2,\dots \\
     &\yvec_{k+1} = \hat{\yvec}_{k+1} / \|\hat{\yvec}_{k+1}\|_1.
\end{aligned}
\end{align}
Note that, for $\alpha=1$ and $M=P$ in \eqref{eq:map} we retrieve exactly the rooted PageRank diffusion \eqref{eq:pr_link}. 
 Unlike the linear case, the convergence of the second order nonlinear process \eqref{eq:us_link} is not straightforward. 
  However, 
  ideas from the proof of Theorem \ref{thm:theory} can be 
  used to show that the convergence is guaranteed for any choice of the tensor $\ten T$, of the matrix $M$, and of the starting point $\yvec_0\geq 0$, provided that the  graph $G_\map$ is aperiodic. 

\begin{figure}[t]
    \centering
    \includegraphics[width=.275\textwidth,trim=0 -2em 0 0]{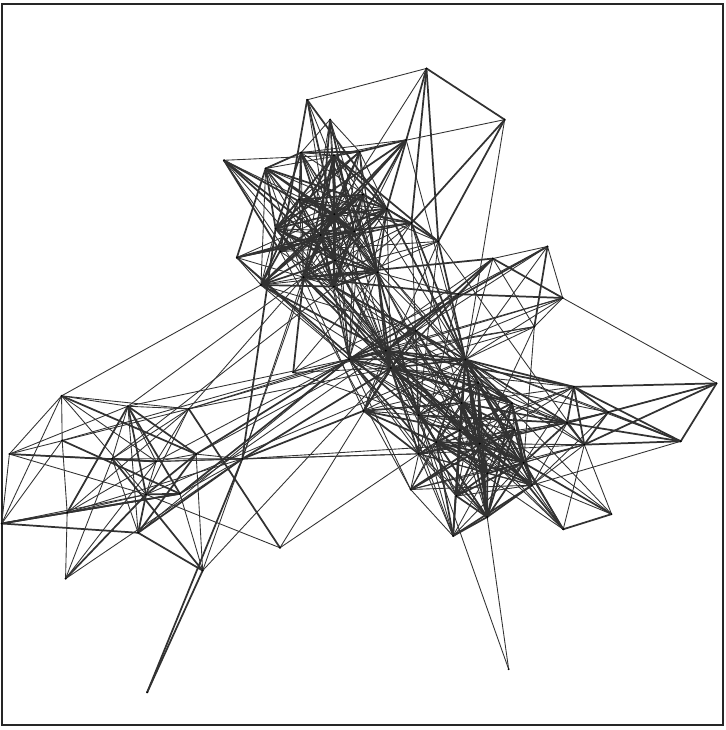}\hspace{2em}
    \includegraphics[width=.31\textwidth]{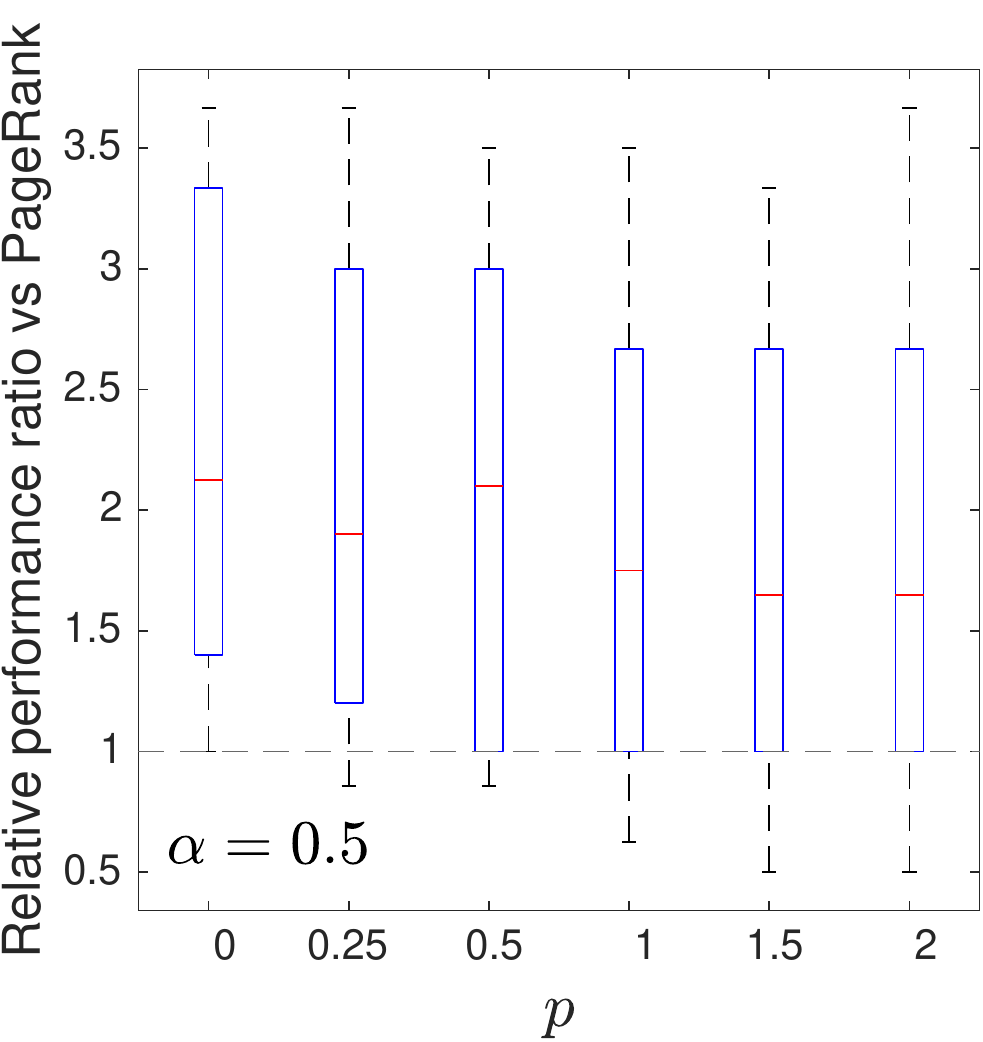}\hspace{.5em}
    \includegraphics[width=.29\textwidth,trim=0cm -.5em 0cm 0cm, clip]{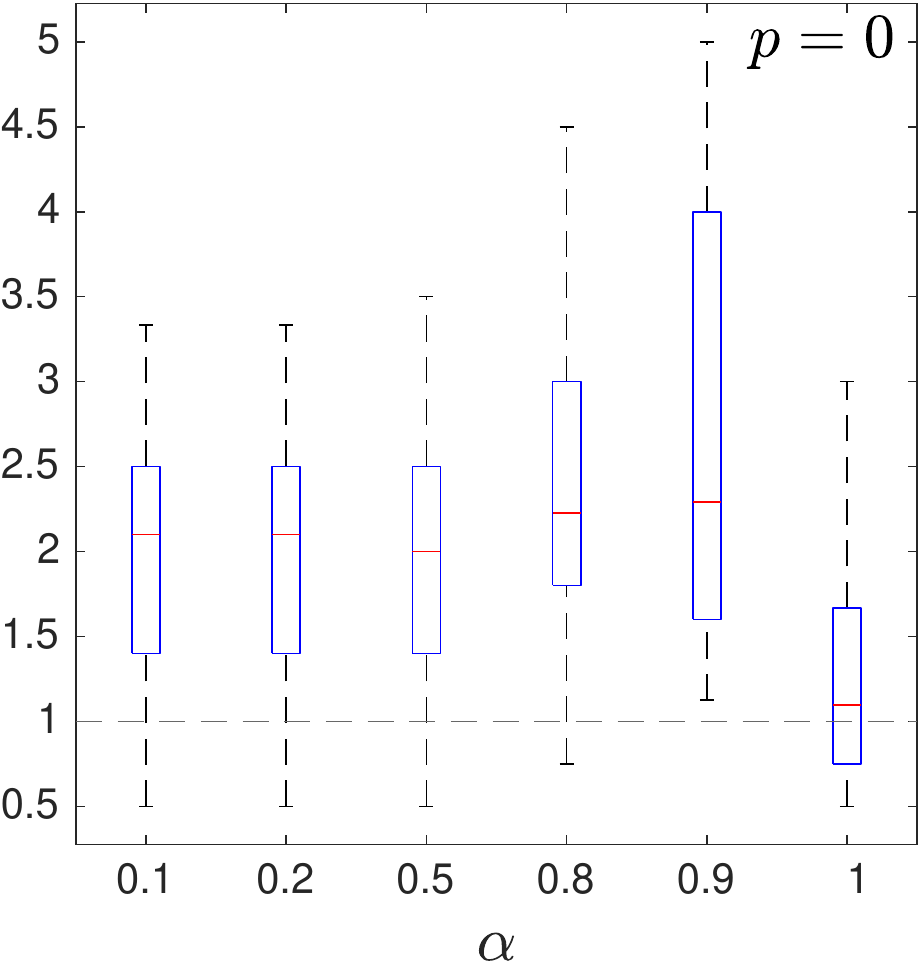}\\
    \includegraphics[width=.276\textwidth,trim=-.4em -2em .4em 0]{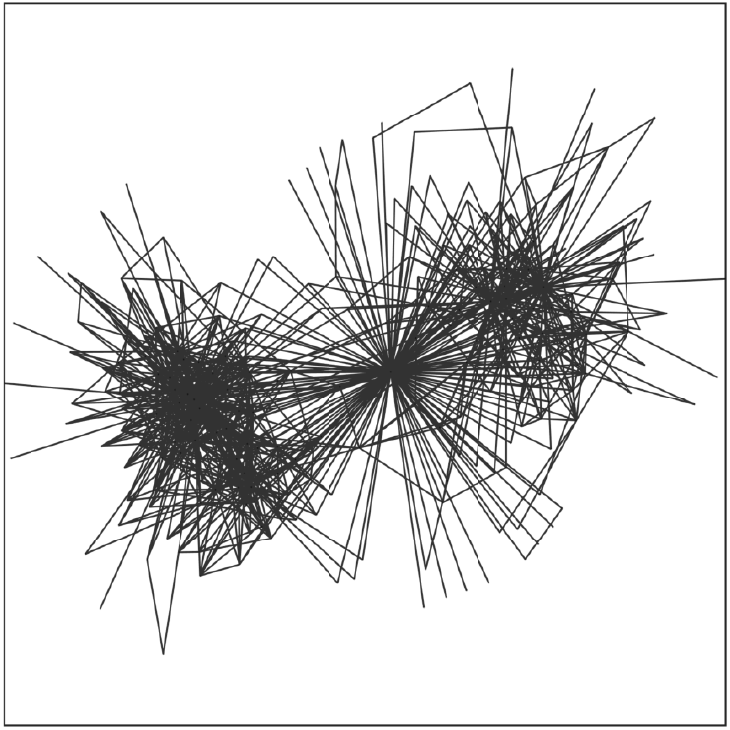}\hspace{2em}
    \includegraphics[width=.31\textwidth]{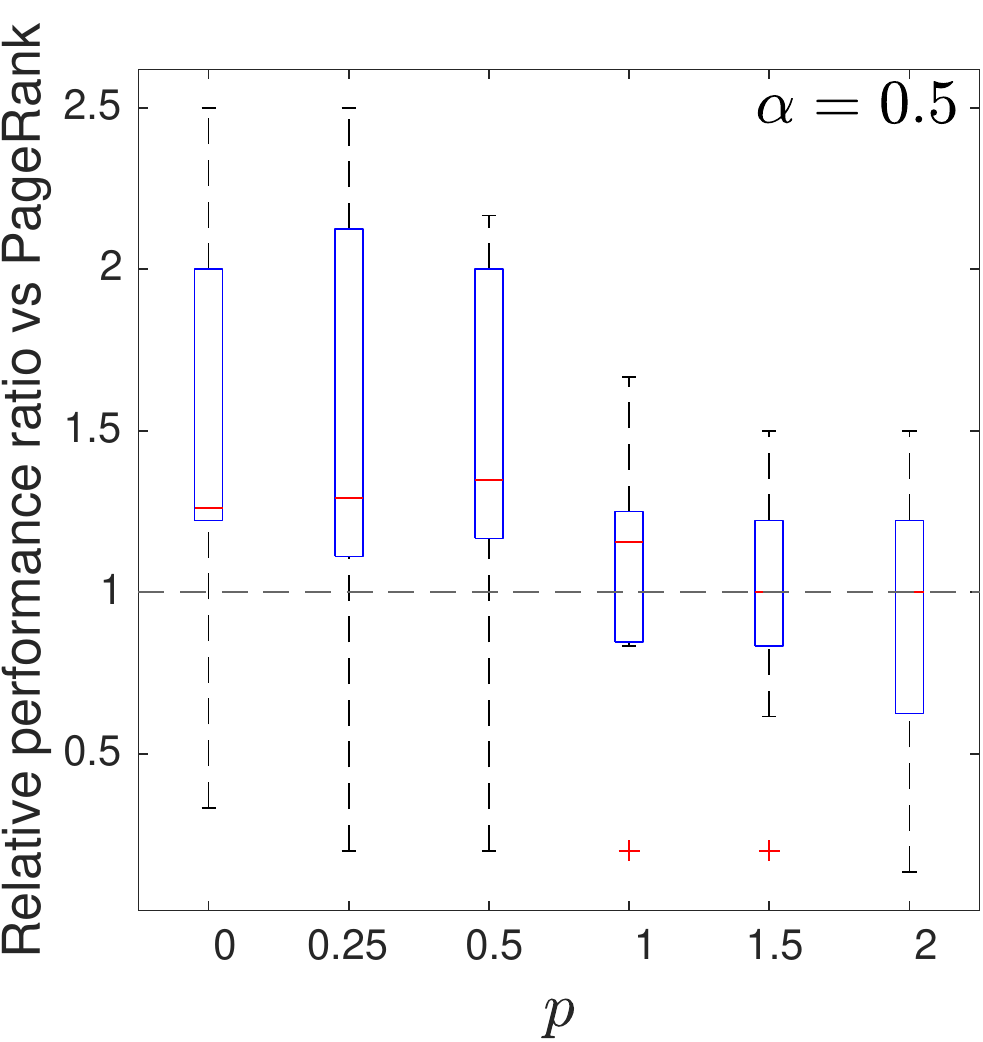}\hspace{.5em}
    \includegraphics[width=.305\textwidth,trim=2.7cm -.1cm -.3cm 0.5cm,clip]{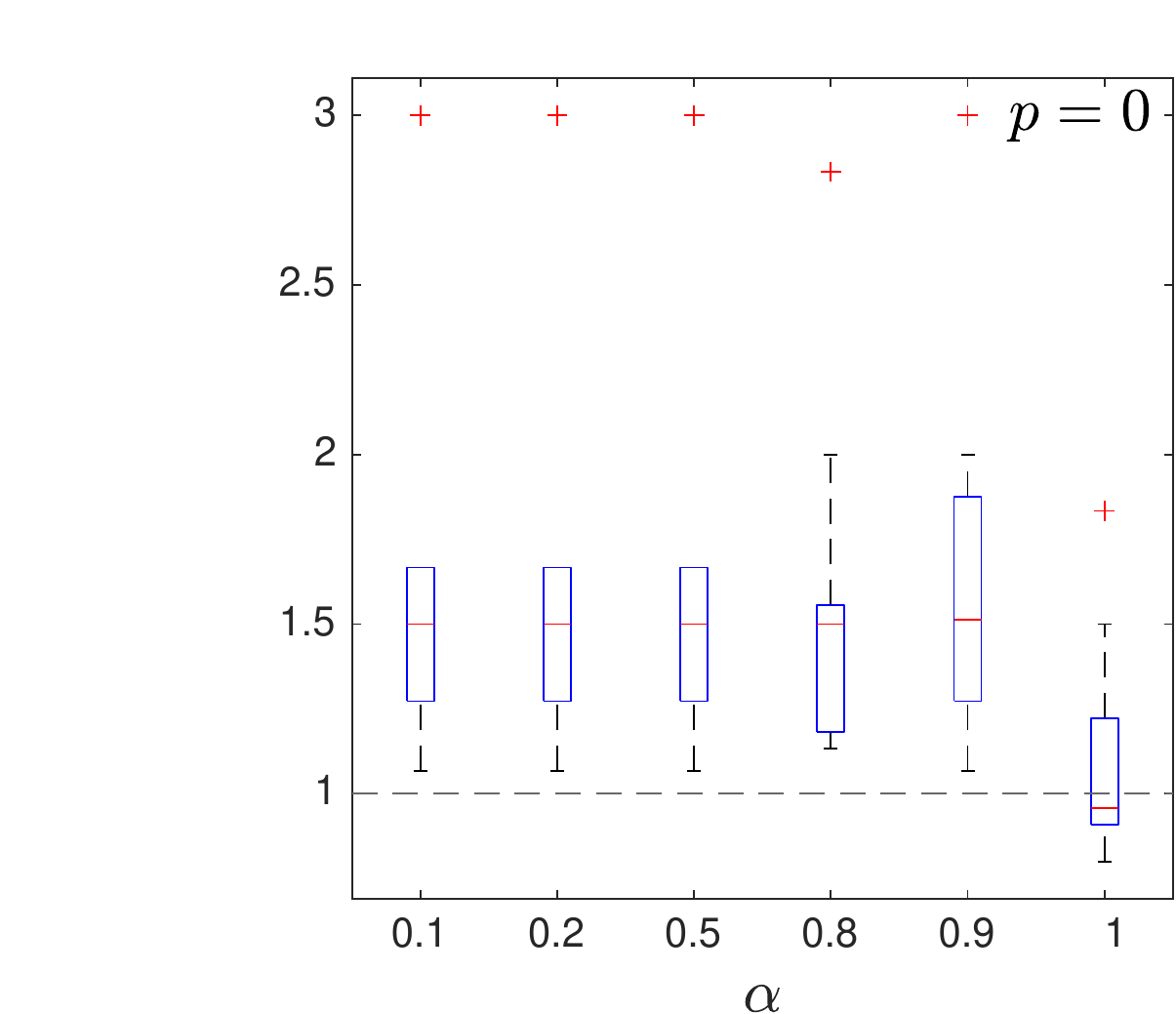}
    \caption{Link prediction performance comparison on two network dataset: {\sc UK faculty} dataset (top) and {\sc Small World citation} network (bottom). The plots show means and quartiles of the ratio between the fraction of correctly predicted edges using $S_\map$ and the one obtained using $S_{PR}$, over ten random trials for different values of $p$ and $\alpha$~in~\eqref{eq:map}. }
    \label{fig:linkprediction}
\end{figure}

\begin{corollary}
 Let $\map:\RR^n\to \RR^n$ be  as in Definition \ref{def:map} and let $G_\map$ be its adjacency graph, as per Definition \ref{def:A_M}. If $G_\map$ is aperiodic and $\yvec_0> 0$, then  $\yvec_k$ defined in \eqref{eq:us_link} for a given seed $\ell$ converges to a unique stationary point $\yvec^{(\ell)}> 0$.
\end{corollary}
\begin{proof}
Let $\mathcal F:\RR^n\to\RR^n$ be the map $\mathcal F(\yvec) = c \map(\yvec) + (1-c)\|\yvec\|_1 \bone_\ell$, where we have omitted the dependency of the map on $\ell$ for the sake of simplicity. 
Note that the limit points of \eqref{eq:us_link} coincide with the fixed points of $\mathcal F$ on the unit sphere $\|\yvec\|_1 = 1$. Note moreover that $\mathcal F$ is homogeneous, i.e.,\ $\mathcal F(\theta\yvec) = \theta \mathcal F(\yvec)$, for all $\theta>0$. Finally, notice that the $j$-th column of the Jacobian matrix of $\mathcal F$ evaluated at $\zvec$ is 
$$
\frac{\partial}{\partial y_{j}}\mathcal F(\zvec) = c\frac{\partial}{\partial y_{j}} \map(\zvec) + (1-c)\bone_\ell \, ,
$$
which shows that if the Jacobian of $\map$ is irreducible, the same holds for the 
Jacobian of $\mathcal F$. With these observations, the  thesis follows straightforwardly using the same arguments as in the proof of Theorem \ref{thm:theory}, applied to $\mathcal F$.
\end{proof}
As for the linear dynamical process, the stationary distributions of \eqref{eq:us_link} computed for different seeds allow us to define the similarity matrix $S_\map$:  
\[(S_\map)_{ij}  = (\yvec^{(i)})_j + (\yvec^{(j)})_i.\] 

In Figure \ref{fig:linkprediction} we compare the performance of the link prediction algorithm based on the standard seeded PageRank similarity matrix $S_{PR}$ \eqref{eq:pr_link} and the newly introduced similarity matrix $S_\map$ \eqref{eq:us_link} induced by $\map$ with $M=P$ and $\ten T = \ten T_W$, the random walk triangle tensor.   
 The tests were performed on the real-world networks {\sc UK faculty} and {\sc Small World citation}. The network {\sc UK faculty}\cite{PhysRevE.77.016107} represents the personal friendships network between the faculty members of a UK University.  The network contains $n=81$ vertices and $m = 817$ edges. 
 The network {\sc Small  World citation} \cite{garfield2004histcite} represents citations among papers that directly cite Milgram's small world paper  or contain the words ``Small World'' in the title. This network contains $n=233$ nodes
and $m=994$ edges. We transformed both networks by neglecting edge directions and weights.
 
The experiments were performed as follows. 
We start with an initial network $G=(V,E)$ and we randomly select a subset of its edges, which we call $E_1$, of size $|E_1|\approx |E|/10$. We then define $G_0=(V,E_0)$ to be the graph obtained from $G$ after removal of the edges in $E_1$, so that $E_0=E\setminus E_1$. Thus, working on the adjacency matrix of $G_0$, we build the two similarity matrices $S_{PR}$ and $S_\map$. Then, for each similarity matrix $S$, we select from $V\times V \setminus E_0$ the subset $E_S$ containing the $|E_1|$ edges with the largest similarity scores $S_{ij}$. 
A better performance corresponds to a larger size of $E_1\cap E_S$, since this is equivalent to detecting more of the edges that were originally in the graph. 
To compare the performance of the two similarity matrices, we thus computed the ratio $|E_{S_\map}\cap E_1|/|E_{S_{PR}}\cap E_1|$. In Figure~\ref{fig:linkprediction} we boxplot this quantity over 10 random runs  where $E_1$ is sampled from the initial $E$ with a  uniform probability.  Whenever the boxplot is above the threshold of 1, our method is outperforming standard seeded PageRank. 
The middle plots in the figure display the results for the two networks when $\alpha=0.5$ in \eqref{eq:map} and we let $p$ vary. On the other hand, the plots on the right display results for varying values of $\alpha$ and $p=0$, which was observed to achieve the best performance in the previous test. 
Overall, the link prediction algorithm based on the similarity matrix $S_\map$ typically outperforms the alternative based on $S_{PR}$, especially for small values of $p$.

\section{Conclusion}
\label{sec:conc}
After associating a network with its adjacency matrix,
it is a natural step   
to 
formulate eigenvalue problems that 
quantify nodal characteristics. 
In this work we 
showed that cubic tensors can be used 
to create a corresponding set of nonlinear eigenvalue problems
that build in higher order effects; notably triangle-based motifs.  
Such spectral measures 
automatically incorporate the mutually reinforcing 
nature of eigenvector and PageRank centrality.
As a special case, we specified
a mutually reinforcing version of the 
classical Watts--Strogatz clustering coefficient.

We showed that our  general framework 
includes a range of approaches for combining 
first and second order  
interactions, and, for all of these, we gave existence and uniqueness results along with 
an effective computational algorithm.
Synthetic and real networks were used to 
illustrate the approach.

Given the recent growth in activity around higher order network features
\cite{BASJK18,B19,benson2015tensor,chodrow2019configuration,EG19,huang2019higher,IPBL19,lambiotte2019networks,li2019motif,tudisco2017node,yin2018,yin2019local,yin2017},
there are many interesting directions in which 
this work could be further developed, including  
the design of such 
centrality 
measures for weighted, directed and dynamic networks, and the study of 
mechanistic network growth models that incorporate higher-order information.

\section*{Data and funding}
The data used in this work is available in the public domain, as indicated in the text. The code used in the experiments is available at \url{https://github.com/ftudisco/SpectralClCoeff}. The three authors contributed equally to the manuscript. The work of FA was supported by fellowship ECF-2018-453 from the Leverhulme Trust.  The work of DJH was supported by EPSRC/RCUK Established Career Fellowship EP/M00158X/1 and by EPSRC Programme Grant EP/P020720/1. The work of FT was partially supported by INdAM--GNCS  and by EPSRC Programme Grant EP/P020720/1.


\end{document}